\documentclass[aps,floatfix,prd,onecolumn,a4paper,11pt,citeautoscript,superscriptaddress,showpacs,tightenlines,notitlepage,nofootinbib]{revtex4-1}

\usepackage{amsmath,amsfonts,amssymb,amsthm}
    \usepackage{dsfont}

\usepackage[colorlinks=true, pdfstartview=FitV,linkcolor=blue,
			citecolor=blue, urlcolor=blue]{hyperref}

\date{\vspace{-1.0cm}\small\today}

\newcommand{\e}{\mathrm{e}}
\newcommand{\imag}{\mathrm{i}}

\renewcommand{\d}{\text{d}}
\newcommand{\dd}{\text{\emph{d}}}

\newcommand{\half}{\mbox{$\frac{1}{2}$}}
\newcommand{\phalf}{\mbox{$ \frac{\pi}{2}$}}
\newcommand{\quarter}{\mbox{$\frac{1}{4}$}}

\newcommand{\ie}{\emph{i.e.}  }
\newcommand{\eg}{\emph{e.g.}  }

\newcommand{\N}{\mathbb{N}}
\newcommand{\Z}{\mathbb{Z}}
\newcommand{\R}{\mathbb{R}}
\newcommand{\C}{\mathbb{C}}

\newtheorem{theorem}{Theorem}[section]
\newtheorem{lemma}[theorem]{Lemma}

\theoremstyle{definition}

\theoremstyle{remark}
\newtheorem{remark}[theorem]{Remark}

\newcommand{\ee}{{\'e} }
\newcommand{\A}{\kappa}
\renewcommand{\L}{}

\newcommand{\wone}{\omega_{1}}
\newcommand{\pwone}{\mbox{$\frac{\pi_{\phantom{1}}}{\wone}$}}
\newcommand{\ptwone}{\mbox{$\frac{\pi}{2 \wone}$}}
\newcommand{\wthr}{\omega_{3}}

\newcommand{\gn}{g_{0}}
\newcommand{\pn}{p_{0}}
\newcommand{\numb}{N}
\newcommand{\kk}{\hat{k}}

\begin{document}    
\title{ \bf Integral representation of solution to the non-stationary Lam{\'e} equation}
\author{Farrokh Atai}
\date{\today}
\email[Electronic address: ]{farrokh@kth.se}
\affiliation{Department of Theoretical Physics, KTH Royal Institute of Technology,  SE-106 91 Stockholm, Sweden}
\begin{abstract}
We consider methods for constructing explicit solutions of the non-stationary Lam\ee equation, which is a generalization of the classical Lam\ee equation, that has appeared in works on integrable models, conformal field theory, high energy physics and representation theory.
We also present a general method for constructing integral representations of solutions to the non-stationary Lam\ee equation by a recursive scheme. Explicit integral representations, for special values of the model parameters, are also presented.
Our approach is based on kernel function methods which can be naturally generalized to the non-stationary Heun equation. 


{\bf Keywords:} non-stationary Lam\ee equation; kernel functions; solutions method; iterative integral representations; 
\end{abstract}

\maketitle
\section{Introduction}

This paper considers a complementary approach to our recent paper \cite{AtaiLangmann2016_2}, which was devoted to the study of a certain non-stationary Schr{\"o}dinger equation with elliptic potentials (see Eq. \eqref{eq_non_stationary_heun_defining}) that have appeared in ongoing research in mathematics and physics; see \eg Refs.
\cite{etingof1994,etingof1995,Felder1995,FELDER1997,FelderVarchenko,EtingofKirillov:1,BazhanovMangazeev:1,BazhanovMangazeev:2,Fateev2009,Rosengren2013,Rosengren2014,Rosengren2015,Kolb2015} (some of which are discussed below). This so-called \emph{non-stationary Heun equation} \cite{Langmann_Takemura,Rosengren2013,Kolb2015}, also known as \emph{quantum Painlev\ee VI} \cite{Nagoya2011}, is a generalization of the Heun equation which is a second order Fuchsian differential equation with four regular singular points. 
The simplest non-trivial special case is the non-stationary Lam\ee equation, also known as the \emph{KZB heat equation} \cite{FelderVarchenko}, which can be written as
\begin{equation}
\Bigl( \frac{2\pi \imag }{\wone^{2}} \A \frac{\partial}{\partial \tau} -\frac{\partial^{2}}{\partial x^{2}} + g \left( g - 1 \right) \wp( x |\half  \wone , \half \wone \tau) \Bigr)  \psi(x,\tau)  = E(\tau) \psi(x,\tau) 
\label{eq_non_stationary_lame_defining}
\end{equation}
with model parameters $(\A,g)$ and $\wp( x |\half  \wone , \half \wone \tau)$ the Weierstrass elliptic function with periods $(\wone , \wone \tau)$ (the definitions of well-known special functions, needed in this paper, are collected in Appendix \ref{appendix_special_functions} for the convenience of the reader). The solutions of the non-stationary Lam\ee equation has also appeared in many different contexts, such as partition functions of the eight-vertex model \cite{BazhanovMangazeev:1,BazhanovMangazeev:2}, exact four-point correlation functions of the quantum Liouville model \cite{Fateev2009}, high energy physics \cite{Felder1995,Falceto1997,FelderVarchenko}, and representation theory of affine Lie algebras \cite{etingof1994,etingof1995,EtingofKirillov:1}. We therefore believe that the results in this paper  contribute to these subjects. 

\subsection{Sketch of solutions}\label{section_sketch_of_solutions}

The aim of this paper is to develop a recursive scheme for constructing functions\footnote{We will suppress the dependence on the parameter $\tau$ in the rest of this paper in order to simplify notation.} \newline $\psi_{n,g}(x)~\equiv~\psi_{n}( x , \tau ; g , \A)$, labeled by integers $n$ and depending on two (complex) variables $x$ and $\tau$, with $\Im(\tau)>0$, which satisfies the non-stationary Lam\ee equation \eqref{eq_non_stationary_lame_defining} for some $E_{n,g} \equiv E_{n}(\tau ; g , \A)$. The solutions constructed in this Paper are of the form
\begin{equation}
\psi_{n,g}(x) = ((2q^{\quarter})^{-2} \vartheta_{1}(\pwone x)^{2})^{\half g} \mathcal{P}_{n,g}(\cos(\pwone x))\quad ( x\in [-\wone , \wone ], n\in\Z , g > -\half)
\end{equation}
with $\vartheta_{1}$ the odd Jacobi theta function \cite{WhitWat}, $\mathcal{P}_{n,g}$ an analytic function, and
\begin{equation}
E_{n,g} = \left( \frac{\pi}{\wone}\right)^{2}( n + g)^{2} - 4g(g-1)\frac{\eta_{1}}{\wone}+ 6g^{2}\bigl( \frac{\eta_{1}}{\wone}-\frac{\pi^{2}}{12\wone^{2}}\bigr),
\label{eq_energy_eigenvalues}
\end{equation}
with $\eta_{1}\equiv \eta_{1}(\tau)$ in \eqref{eq_eta_constant}. A precise characterization of the solutions can be found in \eqref{eq_ns_lame_eigenfunctions}-\eqref{eq_ns_lame_polynomial_expansion}. 

It is important to stress that our solutions depend in an essential way on the parameter $\A$ (as discussed below), and the iterative integral representations are well-defined for real $\A$. (See Section \ref{section_conclusions} for a discussion regarding complex $\A$.) 
(Note that $E$ in \eqref{eq_non_stationary_lame_defining} is of lesser importance when $\A\neq 0$ and $\tau$ finite ($\Im(\tau)>0$), as it can be transformed to $0$ by changing ''normalization'' of $\psi_{n,g}(x)$ by an analytic function of $\tau$. We find it convenient to include $E$ for reasons that is made clear below.)
In particular, we focus on a recursive scheme that allows the solutions of \eqref{eq_non_stationary_lame_defining} to be written as explicit multidimensional integrals. More specifically, we construct explicit integral transforms $\mathcal{K}_{g+\A,g}$  
that maps solutions of the non-stationary Lam\ee equation \eqref{eq_non_stationary_lame_defining} for parameters $(\A ,g )$ to solutions with parameters $(\A, g + \A)$: Let $\psi_{n,g}(x)$ be a solution of \eqref{eq_non_stationary_lame_defining} as described above, then 
\begin{equation}
\psi_{n+s,g+\A}(x ) = (\mathcal{K}_{g+\A,g}\psi_{n,g})(x) \quad (n,s \in \Z)
\label{eq_integral_transform_general}
\end{equation}
with $s$ an integer shift determined by the integral transform operator,\footnote{In Section \ref{section_integer couplings} the shift is proportional to $\A$ and in Section \ref{section_arbitrary_couplings} it is $0$.} and an explicitly known $E$, as is shown in Sections \ref{section_integer couplings} and \ref{section_arbitrary_couplings}. 
(We note that Eq. \eqref{eq_integral_transform_general} suggests that there are unique solutions of the non-stationary Lam\ee equation and we discuss the uniqueness of our solutions in Remark \ref{remark_uniqueness}.)
Using the integral transform in \eqref{eq_integral_transform_general} allows us to construct solutions of the non-stationary Lam\ee equation \eqref{eq_non_stationary_lame_defining} for parameters  $(\A, g ) = (\A,  \numb \A + \gn)$, where $\numb\in \N$ and $0 \leq \gn < \A$ if $\A>0$, or $\gn> - \numb \A$ if $\A < 0$, by repeated application of the integral transform in order to obtain solutions: The iterative integral representations of the solutions to the non-stationary Lam\ee equation are then given as
\begin{equation}
\psi_{n, \numb \A + \gn}(x) = ( \mathcal{K}_{\numb\A + \gn , (\numb-1)\A +\gn}  \mathcal{K}_{(\numb-1)\A + \gn , (\numb-2)\A +\gn} \cdots \mathcal{K}_{\A + \gn , \gn} \psi_{ \tilde{n}, \gn})(x),
\label{eq_recursive_integral_solutions}
\end{equation}
where $\tilde{n}$ is an integer and $\psi_{\tilde{n},\gn}$ a solutions of \eqref{eq_non_stationary_lame_defining} for parameters $(\A, g) = (\A , \gn)$. (The sketched arguments above are made precise in Sections \ref{section_integer couplings} and \ref{section_arbitrary_couplings}.) We often refer to the function $\psi_{n,\gn}$ used in \eqref{eq_recursive_integral_solutions} as the \emph{seed function.}

The seed functions can be constructed, for example, by the series solutions in \cite{AtaiLangmann2016_2} or by the integral representations of the irregular\footnote{As explained in Section \ref{section_conclusions}.} solutions in \eg \cite{etingof1994}. It is, in general, not a simple task to construct a simple, explicit seed function for the non-stationary Lam\ee equation.  However, in the special cases where $\gn = \pn \in \{0,1\}$ then \eqref{eq_non_stationary_lame_defining} reduces to an inhomogeneous variant of the heat equation, given by
\begin{equation}
\Bigl( \frac{2\pi \imag}{\wone^{2}} \A \frac{\partial }{\partial \tau} - \frac{\partial^{2}}{\partial x^{2}} \Bigr) \psi_{n,\pn}(x) = E_{n,\pn}\psi_{n,\pn}(x),
\label{eq_heat_eq}
\end{equation}
and it is straightforward to construct solutions $\psi_{n,\pn}$ of \eqref{eq_heat_eq}; see Sections \ref{section_integer couplings}. The solutions of the non-stationary Lam\ee equation \eqref{eq_non_stationary_lame_defining}, for suitable $\A$ (as explained in Section \ref{section_integer couplings}), are then given by explicit integrals. We will use the following simple special case in order to illustrate the integral transform method for constructing solutions of \eqref{eq_non_stationary_lame_defining}: let $n,\A$ be integers and denote by $\vartheta_{1}(x)$ the odd Jacobi theta function (see Section \ref{section_notation}), then the function\footnote{The integration contour $\mathcal{C}_{\varepsilon}$ is a simple straight line, from $-\wone$ to $\wone$, shifted in the complex plane by $\varepsilon>0$. It is also shown in Section \ref{section_integer couplings} that the function is analytic and independent of the complex shift.}
\begin{equation}
\varphi_{n}(x) \equiv  \int\limits_{\mathcal{C}_{\varepsilon}} \frac{\d y}{2\wone}  \frac{ \vartheta_{1}'(0)^{\A}(\vartheta_{1}(\pwone x)^{2})^{\half\A}}{\Bigl( \vartheta_{1}(\ptwone ( x +y )) \vartheta_{1} ( \ptwone ( x - y)) \Bigr)^{\A} } \cos( \pwone( n + \A) y)
\end{equation}
is a solution of \eqref{eq_non_stationary_lame_defining} for coupling $(\A , g ) = (\A , \A)$, \ie
\begin{equation}
\Bigl( \frac{2\pi \imag }{\wone^{2}} \A \frac{\partial}{\partial \tau}  - \frac{\partial^{2}}{\partial x^{2}} + \A(\A - 1) \wp(x) -E_{n} \Bigr) \varphi_{n}(x) = 0 \quad (n,\A \in \Z)
\label{eq_simple_case}
\end{equation}
with $E_{n} = (\pi/\wone)^{2} (n+\A)^{2}- 4 \A(\A-1) \frac{\eta_{1}}{\wone}$ and $\eta_{1}$ in \eqref{eq_eta_constant}. The results above is also a special case of the results discussed in \cite{Langmann_Takemura}. 
We show in Section \ref{section_conclusions} that the integral transforms are not always unique, in the sense that there exists distinct integral kernels that yield (seemingly different) solutions of \eqref{eq_non_stationary_lame_defining}. We can therefore construct many different integral representations of solutions (see discussion in Section \ref{section_conclusions}) and will thus focus on particular solutions: 
the non-stationary Lam\ee equation \eqref{eq_non_stationary_lame_defining} reduces to an eigenvalue equation for a Schr{\"o}dinger operator, with the P{\"o}schl-Teller potential $\propto \sin(\pi  x / \wone)^{-2}$, in the\emph{ trigonometric limit }$\Im(\tau) \to +\infty$, which has well-known eigenfunctions given by the Gegenbauer polynomials; 
We construct the particular solutions $\psi_{n,g}(x)$ of \eqref{eq_non_stationary_lame_defining} such that they reduce to these well-known eigenfunctions in Section \ref{section_trig_case}, and it is therefore natural to interpret $\psi_{n,g}(x)$ as an elliptic generalization of these eigenfunctions. 

Before we proceed to elaborate on our statements above, we wish to discuss the relations with the results in our paper \cite{AtaiLangmann2016_2} where we constructed solutions of the non-stationary Heun equation using kernel function methods. The solutions of the non-stationary Lam\ee equation, in this paper, corresponds to the special case of the non-stationary Heun equation where all the Heun model parameters $g_{\nu}$ are set to $g$ (see Section \ref{section_conclusions}). In \cite{AtaiLangmann2016_2}, the kernel function is used for constructing a non-trivial set of unconventional basis functions for the non-stationary Heun equation, with particular model parameters (as discussed below). The basis functions are then used in order to solve the non-stationary Heun equation by a perturbative algorithm and the kernel function method in \cite{AtaiLangmann2016_2} reduces solving a partial differential equation (see \eqref{eq_non_stationary_heun_defining}) to finding coefficients satisfying a particular differential-difference equation (see Section 3 of \cite{AtaiLangmann2016_2}). 
In this perturbative construction there appears to be problems with resonance, \ie vanishing generalized energy differences, and the Heun model parameters are restricted by, what we refer to as, \emph{no-resonance conditions} (see also Remark \ref{remark_resonance_cases}). The resonance cases seem to be most prevalent when the Heun model parameters take integer values, and they also occur in the non-stationary Lam\ee case (see Remark \ref{remark_resonance_cases}). 
The iterative integral approach is especially suitable for integer parameters as the iterative integrals becomes more elementary. We therefore give a particular focus on the integer case in Section \ref{section_integer couplings} before considering the arbitrary coupling case in Section \ref{section_arbitrary_couplings}.

\subsection{Notation}\label{section_notation}
In this section we collect some of the common notations and definitions that are used in this paper.
In order to simplify notations in the proceeding discussions we can set $$\wone \equiv \pi$$ without loss of generality. (See Appendix B of \cite{AtaiLangmann2016_2} on how to re-introduce arbitrary $\wone$.)
We use standard notations \cite{WhitWat} for the elliptic \emph{nom\ee} 
\begin{equation}q \equiv \exp( \imag \pi \tau)\end{equation} and $\eta_{1}$,
\begin{equation}
\frac{\eta_{1}}{\pi} \equiv \frac{1}{12} - 2 \sum_{n\in\N} \frac{q^{2n}}{(1-q^{2n})^{2}}.
\label{eq_eta_constant}
\end{equation}
We also define 
\begin{equation}
G \equiv \prod_{n\in\N}( 1 - q^{2n} ),
\label{eq_G_constant}
\end{equation}
where it is straightforward to check that 
\begin{equation}
\frac{\imag}{\pi} \frac{1}{G} \frac{\partial G}{\partial \tau} = -\frac{1}{G}q\frac{\partial G}{\partial q} = \frac{1}{12} - \frac{\eta_{1}}{\pi}
\end{equation}
which is needed as we proceed.

We define the $\theta_{\nu}$-functions ($\nu=1,2,3,4$) as
\begin{equation}\begin{split}
\theta_{1}(x) \equiv G \sin(x) \prod\limits_{n\in \N}( 1 - 2 q^{2n} \cos(2x) + q^{4n }) ,& \quad \theta_{2}(x) \equiv G \cos(x) \prod\limits_{n\in\N}( 1 + q^{2n} \cos(2x) + q^{4n}) \\
\theta_{3}(x) \equiv G \prod\limits_{n\in\N}( 1 + 2 q^{2n-1}\cos(2x) + q^{4n-2} ) , &\quad \theta_{4}(x) \equiv  G \prod\limits_{n\in\N}( 1 - 2 q^{2n-1}\cos(2x) + q^{4n-2} ) .
\end{split}\label{eq_Jacobi_theta_functions}\end{equation}
(Note that this notation deviates from our notation in \cite{AtaiLangmann2016_2}.)
We would also like to point out that the $\theta_{\nu}$-functions are essentially the Jacobi $\vartheta_{\nu}$-functions \cite{WhitWat} (see Appendix \ref{appendix_trig_limit}): $\theta_{\nu} = (2q^{\quarter})^{-1} \vartheta_{\nu}$ for $\nu=1,2$ and $\theta_{\nu}=\vartheta_{\nu}$ for $\nu=3,4$.
We also define the functions
\begin{subequations}
\begin{align}
\Theta_{1}(\xi) &\equiv ( 1 - \xi) \prod_{l \in \N} ( 1 - q^{2l} \xi) ( 1 - q^{2l} \xi^{-1}) \\
\Theta(z, \xi) &\equiv ( 1 - 2 z \xi + \xi^{2} ) \prod_{l \in \N} ( 1 - 2q^{2l} z  \xi + q^{4l} \xi^{2})( 1 - 2 q^{2l} z \xi^{-1} + q^{4l} \xi^{-2})
\end{align}
\end{subequations}
which are closely related to the $\theta_{1}$ function (see Appendix \ref{appendix_special_functions}).  

We denote by $\N_{\geq\A}$ the set of all positive integers larger than, or equal to, $\kappa$. (Recall that we only consider the cases where $\A$ is real.)

\subsection{Trigonometric limit}\label{section_trig_case}

To set the results of the next Section in perspective, we consider the non-stationary Lam\ee equation in the trigonometric limit $\Im(\tau) \to + \infty$. It follows from known  relations (see Appendix \ref{section_trig_case}) that the non-stationary Lam\ee equation in \eqref{eq_non_stationary_lame_defining} reduces to the eigenvalue equation 
\begin{equation}
\Bigl( -\frac{\partial^{2}}{\partial x^{2}} + g ( g -1 ) \frac{1}{\sin( x )^{2}} \Bigr) \psi^{(0)}_{n,g}(x) = E_{n,g}^{(0)} \psi^{(0)}_{n,g}(x),
\label{eq_Poschl_Teller_Hamiltonian}
\end{equation}
in the trigonometric limit. 
The Schr{\"o}dinger operator in \eqref{eq_Poschl_Teller_Hamiltonian} is a special case of a Hamiltonian operator with the P{\"o}schl-Teller potential, and has well-known eigenfunctions given by
\begin{equation}
\psi_{n}^{(0)}(x;g) \equiv (\sin(x)^{2})^{\half g} C_{n}^{(g)}( \cos(x) ) \quad (n\in \N_{0}),
\label{eq_Poschl_Teller_eigenfunctions}
\end{equation}
where $C_{n}^{(g)}$ are the \emph{Gegenbauer polynomials},\footnote{Also known as the \emph{ultra-spherical polynomials}.} and $ E_{n,g}^{(0)} \equiv (n + g)^{2}$ the corresponding eigenvalues.

The functions in \eqref{eq_Poschl_Teller_eigenfunctions} are the unique (non-zero) normalizable solutions of a particular self-adjoint extension of the Schr{\"o}dinger operator in \eqref{eq_Poschl_Teller_Hamiltonian}.

\section{Integer coupling}\label{section_integer couplings}

In this section we will mainly focus on the special case that violates the no-resonance conditions of Section 2.4 of our paper \cite{AtaiLangmann2016_2}, as explained in Section \ref{section_sketch_of_solutions}. This entails that we give particular focus on the special case where $(\A,g)$ in \eqref{eq_non_stationary_lame_defining} take (positive) integer values. (See Section \ref{section_conclusions} for a discussion regarding negative couplings.)
We wish to stress that our approach in this Section can be easily extended for suitably complex $x$, but that we choose $\Im(x)=0$ in order to not obfuscate the results below with technicalities.

\subsection{Results}

We construct particular solutions of the non-stationary Lam\ee equation in \eqref{eq_non_stationary_lame_defining} of the form
\begin{equation}
\psi_{n,g}(x) \equiv (\theta_{1}(x)^{2})^{\half g} \mathcal{P}_{n,g}(\cos(x)) \quad ( x \in [-\pi , \pi ] , n \in \N_{0} )
\label{eq_ns_lame_eigenfunctions}
\end{equation}
with $\mathcal{P}_{n,g}(z)$ an analytic function\footnote{A precise definition of an analytic function of several (complex) variables is found in Chapter II of \cite{hormander}. See also Lemma \ref{lemma_Hartogs_theorem}.} of $z$ and the nom\ee $q$, that has the expansion
\begin{equation}
\mathcal{P}_{n,g}(z) = C_{n}^{(g)}(z) + \sum_{\ell=1}^{\infty} \mathcal{P}^{(\ell)}_{n,g}(z) q^{2\ell} \quad (n \in \N_{0})
\label{eq_ns_lame_polynomial_expansion}
\end{equation}
where $C_{n}^{(g)}$ are the Gegenbauer polynomials and $\mathcal{P}_{n,g}^{(\ell)}$ a polynomial of order $n+2\ell$.
We refer to the integer $n$ in $\psi_{n,g}(x)$, and $\mathcal{P}_{n,g}$, as the \emph{degree}, as it corresponds to the polynomial degree in the $q=0$ case.
It follows from \eqref{eq_ns_lame_eigenfunctions}--\eqref{eq_ns_lame_polynomial_expansion}, and \eqref{eq_trig_limit_theta}, that these particular solutions $\psi_{n,g}(x)$ reduce to the eigenfunctions in \eqref{eq_Poschl_Teller_eigenfunctions} in the trigonometric limit $q \to 0$. The significance of the generalized energy eigenvalue then becomes clear as we require a function $E_{n,g}$ in \eqref{eq_non_stationary_lame_defining}, with the expansion $E_{n,g}= (n+g)^{2}+ \mathcal{O}(q)$, in order to obtain the corresponding eigenvalue equation \eqref{eq_Poschl_Teller_Hamiltonian} in the trigonometric limit.

We are now in a position where we can introduce one of the integral transform operators: Let $\A$ be an integer and $\psi_{n,g}(x)$, $x\in [ -\pi , \pi ]$, a function of the form \eqref{eq_ns_lame_eigenfunctions}-\eqref{eq_ns_lame_polynomial_expansion}, then the integral operator $K_{g+\A,g}$ is defined as
\begin{equation}
(K_{g+\A, g} \psi_{n,g})(x) \equiv \mathcal{N}_{n,g,\A} \int_{-\pi + \imag \varepsilon}^{\pi + \imag \varepsilon} \frac{\d y}{2\pi} \frac{(\theta_{1}( x )^{2})^{\half(g+\A)} (\theta_{1}(y)^{2})^{\half g}}{ (\theta_{1}(\half(x+y)) \theta_{1}(\half(x-y)))^{2g + \A}} \psi_{n,g}(y),
\label{eq_integral_transform_1}
\end{equation}
with $0 < \varepsilon < \pi \Im(\tau)$, and normalization constant $\mathcal{N}_{n,g\A}$ given by
\begin{equation}
\mathcal{N}_{n,g,\A} \equiv 2^{2(g+\A)}  \frac{n!}{(2g+\A)_{n-\A} (g)_{\A}},  
\label{eq_norm_constant_1}
\end{equation}
where $(x)_{n}$ denotes the (raising) Pochhammer symbol (see Appendix \ref{appendix_special_functions}). 
The normalization in \eqref {eq_norm_constant_1} ensure that the integral transform preserves the expansion in \eqref{eq_ns_lame_polynomial_expansion}, \ie the coefficient in front of the Gegenbauer polynomial in the expansion of \eqref{eq_integral_transform_1} is $1$. 

As we show below, the integral is well-defined, and independent of $\varepsilon$, if $\psi_{n,g}$ is of the form \eqref{eq_ns_lame_eigenfunctions}-\eqref{eq_ns_lame_polynomial_expansion} and $\A \in \Z\setminus\{0\}$.
In fact, it is shown that the integral transform is well-defined for any $2\pi$-periodic function $\phi(x)$ that is analytic in the strip $\{x\in\C : |\Im(x)|<\pi\Im(\tau)\}$ (recall that $\Im(\tau)>0$).

(Recall the definition of $\N_{\geq\A}$ in Section \ref{section_notation} and $\psi_{n,g}$ in \eqref{eq_ns_lame_eigenfunctions}-\eqref{eq_ns_lame_polynomial_expansion}.)

\begin{theorem}\label{theorem_1}
Let $\A \in \Z$, $n\in\N_{\geq\A}$, $g\geq 0$ be fixed, and $\psi_{n,g}(x)$ be of the form \eqref{eq_ns_lame_eigenfunctions}-\eqref{eq_ns_lame_polynomial_expansion}, as well as a solution of the non-stationary Lam\ee equation \eqref{eq_non_stationary_lame_defining}, for parameters $(\A , g)$, with some $E_{n,g}$. Then the function
\begin{equation}
( K_{g+\A,g} \psi_{n,g})(x)
\label{eq_theorem_integral_transform}
\end{equation}
is a solutions of the non-stationary Lam\ee equation \eqref{eq_non_stationary_lame_defining}, with parameters $(\A,g+\A)$, for 
\begin{equation}
E = E_{n,g}+ 4\A( 1 - 2g -3 \A)\frac{\eta_{1}}{\pi} + 6 \A(2g+\A)\bigl( \frac{\eta_{1}}{\pi} - \frac{1}{12}\bigr).
\end{equation}
Moreover, the function in \eqref{eq_theorem_integral_transform} is an analytic function of $x$ and $\tau$ \emph{(}$\Im(\tau)>0$\emph{)}, independent of $\varepsilon$, and of the form \eqref{eq_ns_lame_eigenfunctions}-\eqref{eq_ns_lame_polynomial_expansion} for coupling $(\A , g+\A)$ and degree $n-\A$.
\end{theorem}
\noindent (The proof is given in Section \ref{section_proof_of_theorem_1}.)

It follows from Theorem \ref{theorem_1} that the integral transform can be considered as a \emph{step operator} and that the functions $\psi_{n,g}$ satisfy 
\begin{equation}
\psi_{n-\A,g+\A}(x) = ( K_{g+\A,g}\psi_{n,g})(x).
\label{eq_recursive_integral_relation_1}
\end{equation}
Following the procedure in \ref{section_sketch_of_solutions} allows us to construct solutions of \ref{eq_non_stationary_lame_defining} by \eqref{eq_recursive_integral_solutions}, with the integral transform in \eqref{eq_integral_transform_1}. The cases where $\gn = \pn \in \{ 0 , 1 \}$ allows for explicit solutions of the non-stationary Lam\ee equation by using the seed functions
\begin{equation}
\begin{split}
\psi_{n,0}(x) &= \cos( n x) \\
\psi_{n,1}(x) &= (\sin( ( n+1)x )^{2})^{\half} / G^{3/\A}
\end{split} \quad (n \in \N_{0}, \A \neq 0).
\label{eq_seed_functions_1}
\end{equation}
\begin{remark}
Note that the functions in \eqref{eq_seed_functions_1} are still of the form \eqref{eq_ns_lame_eigenfunctions}-\eqref{eq_ns_lame_polynomial_expansion} (in the $\pn = 1$ case this follows from
\begin{equation}
(\sin( (n+1)x )^{2})^{\half} = (\theta_{1}(x)^{2})^{\half}  \Bigl(\frac{\sin((n+1)x)^{2}}{\theta_{1}(x)^{2}}\Bigr)^{\half}, 
\end{equation}
where
\begin{equation} \Bigl(\frac{\sin((n+1)x)^{2}}{\theta_{1}(x)^{2}}\Bigr)^{\half} = U_{n}(\cos(x)) + \sum_{\ell\in\N} \mathcal{P}_{n,1}^{(\ell)}(\cos(x)) q^{2\ell}
\label{eq_sine_expansion}
\end{equation}
with $U_{n}(z)$ the Chebyshev polynomials of second kind and $ \mathcal{P}_{n,1}^{(\ell)}(z)$ a polynomial of order $n+2\ell$).
\end{remark}

 Let $\A>0$, $\numb \in \N$, $x\in [-\pi , \pi ]$, and $0 < \varepsilon_{\numb} < \ldots < \varepsilon_{1} < \pi \Im(\tau)$.\footnote{The ordering of follows from Lemma \ref{lemma_KF_analytic} and Eq. \eqref{eq_analytic_region}.}  Then the solutions of the \eqref{eq_non_stationary_lame_defining} for parameters $(\A, \numb \A + \pn )$ has the explicit integral representation
\begin{equation}
\begin{split}
\psi_{n,\numb \A}(x) = (\theta_{1}(x)^{2})^{\half \numb \A} \mathcal{N}_{(0)}\int\limits_{\mathcal{C}_{\varepsilon_{\numb}}} \frac{\d y_{\numb}}{2\pi} \cdots \int\limits_{\mathcal{C}_{\varepsilon_{1}}} \frac{\d y_{1}}{2\pi} \frac{( \theta_{1}(y_{\numb})^{2})^{(\numb-1)\A}}{(\theta_{1}(\half(x+y_{\numb})) \theta_{1}(\half ( x-y_{\numb}) ) )^{(2\numb-1)\A}}  \\ 
\times \prod\limits_{j=1}^{\numb-1} \frac{(\theta_{1}(y_{j})^{2})^{(j-1) \A}}{( \theta_{1}(\half(y_{j+1} + y_{j})) \theta_{1}(\half(y_{j+1}-y_{j})))^{(2j-1)\A}} \cos( (n+\numb\A)y_{1}),
\end{split}
\end{equation}
with normalization constant
\begin{equation}
\begin{split}
\mathcal{N}_{(0)} &\equiv \mathcal{N}_{n+\numb \A,\numb \A , \A} \mathcal{N}_{n+\numb \A,(\numb - 1) \A , \A} \cdots \mathcal{N}_{n+\numb \A, \A , \A} \\  
&= 2^{ \A\numb(\numb+3)} \e^{\imag \pi \half \numb( \numb +1 ) \A} \prod_{j=1}^{\numb} \frac{\Gamma(n + \numb \A + 1)\Gamma((2j+1)\A) \Gamma(j\A)}{\Gamma( n + (N+ 2j) \A )\Gamma((j+1)\A)},
\end{split}
\label{eq_norm_constant_full_0}
\end{equation}
for $\pn = 0$ and
\begin{equation}
\begin{split}
\psi_{n,\numb \A + 1}(x) &= (\theta_{1}(x)^{2})^{\half(\numb \A + 1)} \mathcal{N}_{(1)}\int\limits_{\mathcal{C}_{\varepsilon_{\numb}}} \frac{\d y_{\numb }}{2\pi} \cdots \int\limits_{\mathcal{C}_{\varepsilon_{1}}}  \frac{\d y_{1}}{2\pi}\frac{(\theta_{1}(y_{\numb })^{2})^{(\numb-1)\A + 1}}{( \theta_{1}(\half(x+y_{\numb  })) \theta_{1}(\half(x-y_{\numb })))^{(2\numb - 1 )\A +2} } \\ 
&\times  \prod\limits_{j=1}^{\numb - 1} \frac{(\theta_{1}(y_{j})^{2})^{(j-1)\A+1}}{(\theta_{1}(\half(y_{j+1} + y_{j})) \theta_{1}(\half( y_{j+1} - y_{j})))^{(2j-1)\A+2}}
\Bigl(\frac{\sin( ( n + \numb \A + 1 ) y_{1} )^{2} }{G^{3 / \A} \theta_{1}(y_{1})^{2}}\Bigr)^{\half},
\end{split}
\label{eq_solution_1}
\end{equation}
with normalization constant
\begin{equation}
\begin{split}
\mathcal{N}_{(1)} & \equiv \mathcal{N}_{n+\numb \A,\numb \A +1 , \A} \mathcal{N}_{n + \numb \A , (\numb - 1 ) \A + 1 , \A } \cdots \mathcal{N}_{ n + \numb \A , \A + 1 , \A}\\
&= 2^{\numb((\numb + 3 )\A +2 )} \e^{-\imag \half \pi( \numb (\numb +1)\A+2)} \prod_{j=1}^{\numb} \frac{\Gamma(n+\numb\A+1) \Gamma( (2j+1)\A +2) \Gamma(2\A +1)}{\Gamma(n+(\numb + 2j) \A +2 ) \Gamma((j+1)\A +1)},
\end{split}
\label{eq_norm_constant_full_1}
\end{equation}
for $\pn = 1$. (We note that the representation in \eqref{eq_solution_1} also works for negative parameters $-1 < \A < 0$ and positive integer $\numb$ such that $\numb\A+ 1 >0$.)

It is then also quite straightforward to obtain the following properties of these particular solutions of the non-stationary Lam\ee equation from Theorem \ref{theorem_1}:
\begin{align}
\psi_{n,g}(-x) &= \psi_{n,g}(x) \\
\psi_{n,g}(x + 2\pi) &= \psi_{n,g}(x).
\end{align}
(Theorem \ref{theorem_1} also suggest that $\psi_{n,g}(x+2\pi\tau) = \exp(2\imag \A ( x +\pi \tau) )\psi_{n,g}(x)$ but this we do not prove.)

\subsection{Proof of Theorem \ref{theorem_1}}\label{section_proof_of_theorem_1}
In order to prove Theorem \ref{theorem_1} we proceed with three key Lemmas; we will show that the integral transform satisfies the non-stationary Lam\ee equation in \eqref{eq_non_stationary_lame_defining} (for some $E$). This follows from the \emph{generalized kernel function identity} in Lemma \ref{lemma_gen_KF_id_1}; We then show that the integral transform yields an analytic function if the seed function is analytic and of the form \eqref{eq_ns_lame_eigenfunctions}-\eqref{eq_ns_lame_polynomial_expansion} (see Lemma \ref{lemma_KF_analytic}); Finally, in Lemma \ref{lemma_space_endomorphism} it is shown that the integral transform preserves the general forms in \eqref{eq_ns_lame_eigenfunctions}-\eqref{eq_ns_lame_polynomial_expansion} if $\psi_{n,g}$ is of the same form.  and determine the normalization constant in \eqref{eq_norm_constant_1}.

Theorem \ref{theorem_1} follows then from Lemmas \ref{lemma_gen_KF_id_1}, \ref{lemma_KF_analytic}, and \ref{lemma_space_endomorphism}.

In order to simplify notations below, we introduce the Lam\ee differential operator $H_{\L}(x;g)$, defined as 
\begin{equation}
H_{\L}(x;g) \equiv -\frac{\partial^{2}}{\partial x^{2}} + g(g-1) \wp(x),
\label{eq_Lame_operator}
\end{equation}
which allows us to write the non-stationary Lam\ee equation as
\begin{equation}
\Bigl( \frac{2\imag }{\pi} \A \frac{\partial }{\partial \tau} + H_{\L}(x;g) - E \Bigr)\psi(x) = 0.
\end{equation}
(Recall the definitions of the odd Jacobi theta function in \eqref{eq_Jacobi_theta_functions}). 
The first part of Theorem \ref{theorem_1} follows from the generalized kernel function identity:
\begin{lemma}\label{lemma_gen_KF_id_1}
The {kernel function}
\begin{equation}
\kk(x,y;g,\A) \equiv \frac{(2q^{\quarter})^{-(2g+\A)} \vartheta_{1}'(0)^{2g+\A} \theta_{1}(x)^{g+\A}\theta_{1}(y)^{g}}{( \theta_{1}(\half(x+y)) \theta_{1}(\half(x-y)))^{2g+\A}}
\label{eq_kernel_function}
\end{equation}
satisfies the {generalized kernel function identity}
\begin{equation}
\Bigl( \frac{2\imag}{\pi} \A \frac{\partial }{\partial \tau} + H_{\L}(x;g+\A) - H_{\L}(y; g ) - C \Bigr) \kk(x,y) = 0,
\label{eq_kernel_function_identity}
\end{equation}
with $C = 4 \A ( 1 - 2 g - \A )\frac{\eta_{1}}{\pi}$.
\end{lemma}
\noindent (This is a special case of Corollary 3.2 in \cite{Langmann_Takemura} with $N=M=1$, $\beta=-2\pi \imag \tau$, $A_{1,1} = 4 \A$, and all the couplings set to $g$. We recall the proof in Appendix \ref{appendix_proof_of_lemmas}.)

Note that the kernel of the integral transform in \eqref{eq_integral_transform_1} is the kernel function in \eqref{eq_kernel_function} (up to multiplication by an analytic function of $q$). This allows us to consider the action of the operator $(2\imag/\pi)\A (\partial / \partial \tau) + H_{\L}$ on the function in \eqref{eq_theorem_integral_transform}: Assuming for now that we are allowed to interchange differentiation and integrations, then
\begin{align}
&\Bigl( \frac{2\imag}{\pi}\frac{\partial}{\partial \tau} + H_{\L}(x;g+\A) \Bigr) (K_{g+\A,g}\psi_{n,g})(x) \nonumber \\ 
&= \mathcal{N}_{n,g,\A}G^{-3(2g+\A)} \int_{\mathcal{C}_{\varepsilon}} \frac{\d y}{2\pi} \psi_{n,g}(y) \Bigl( \frac{2\imag}{\pi}\A \frac{\partial}{\partial \tau}+ H_{\L}(x;g+\A) \Bigr)  \kk(x,y)  \nonumber\\ 
&+  \mathcal{N}_{n,g,\A}G^{-3(2g+\A)} \int_{\mathcal{C}_{\varepsilon}} \frac{\d y}{2\pi} \frac{2\imag}{\pi}\A \frac{\partial\psi_{n,g}(y)}{\partial \tau} \kk(x,y) +  \frac{2\imag}{\pi}\A\frac{\partial (\mathcal{N}_{n,g,\A}G^{-3(2g+\A)}) }{\partial \tau} \int_{\mathcal{C}_{\varepsilon}} \frac{\d y}{2\pi} \psi_{n,g}(y) \kk(x,y) \nonumber \\
&=  B.T. + \mathcal{N}_{n,g,\A}G^{-3(2g+\A)} \int_{\mathcal{C}_{\varepsilon}} \frac{\d y}{2\pi} \kk(x,y) ( \frac{2\imag}{\pi}\A \frac{\partial}{\partial \tau} + H_{\L}(y;g) + C_{1} ) \psi_{n,g}(y) 
\end{align}
with $$C_{1} \equiv C + \frac{2\imag}{\pi} \A \frac{1}{G^{-3(2g+\A)}} \frac{\partial G^{-3(2g+\A)}}{\partial \tau} = 4\A(1-2g-\A)\frac{\eta_{1}}{\pi} + 6\A(2g+\A) \Bigl( \frac{\eta_{1}}{\pi} - \frac{1}{12}\Bigr) .$$ 
The last line was obtained by using the generalized kernel function identity and integrating by parts (twice), which yields the boundary terms 
$$B.T. \equiv \int_{\mathcal{C}_{\varepsilon}} \frac{\d y}{2\pi} \frac{\partial}{\partial y} \Bigl( \psi_{n,g}(y) \frac{\partial}{\partial y} \kk(x,y) - \kk(x,y)\frac{\partial}{\partial y} \psi_{n,g}(y) \Bigr).$$
It is a straightforward check\footnote{It follows from the fact that the kernel function $\kk(x,y)$ and $\psi_{n,g}(y)$ are analytic, $2\pi$-periodic functions of $y$, for $y\in \mathcal{C}_{\varepsilon}$, $\varepsilon>0$, and $x\in[-\pi , \pi]$.} that $B.T.=0$ if $\psi_{n,g}$ is an analytic function and of the form \eqref{eq_ns_lame_eigenfunctions}-\eqref{eq_ns_lame_polynomial_expansion} and $\mathcal{C}_{\varepsilon} = [-\pi , \pi] + \imag \varepsilon$ with $0 < \varepsilon < \pi \Im(\tau)$ (recall that we set $\Im(x)=0$).

It follows from the (formal) calculations above that the function in \eqref{eq_theorem_integral_transform} is a solution of the non-stationary Lam\ee equation, for parameter $(\A, g+\A)$, with $E= E_{n,g} + C_{1}$, if $\psi_{n,g}(y)$, and $E_{n,g}$, are solutions of \eqref{eq_non_stationary_lame_defining} for parameters $(\A,g)$.

We turn to making the formal arguments above by showing that the integrand is an analytic function of all three variables $x,y,\tau$. It is also shown, in the proof of Lemma \ref{lemma_space_endomorphism}, that the integrand is an analytic function of $\cos(x)$, $\exp(\imag y)$, and $\exp(\imag \pi \tau)$. We also note that the interchange of integration and differentiation (above), is thus justified if the integrand is an analytic function in the integration domain. 

We recall some well-known results about analytic functions in several complex variable:
\begin{lemma}[Hartogs Theorem]
Let $u(z_{1},z_{2},\ldots,z_{N})$ be a complex valued function defined in the open set $\Omega \subset \C^{N}$. If $u$ is analytic in each variable $z_{j}$, when the other variables are given arbitrary fixed values, then $u$ is analytic in $\Omega$.
\label{lemma_Hartogs_theorem}
\end{lemma}\noindent (Proof can be found in \eg Chapter II of \cite{hormander}. Lemma \ref{lemma_Hartogs_theorem} is also known as Osgood's Lemma when, in addition, the functions are assumed to be continuous.)

It then follows from \eqref{eq_Jacobi_theta_functions}, and \eqref{eq_ns_lame_polynomial_expansion}, that the eigenfunctions $\psi_{n,g}(y)$ in \eqref{eq_ns_lame_eigenfunctions} are analytic functions of $y$ and $\tau$, in the integration domain if $g$ is an integer and the series in \eqref{eq_ns_lame_polynomial_expansion} is absolutely convergent. (Note that for non-integer values of $g$ (and complex $y$) there exists branch cuts along the real axis.) 

We are then left to show that the kernel function is analytic w.r.t. three variables:
\begin{lemma}\label{lemma_KF_analytic}
For $g,\A$ positive integers, the function $\kk(x,y)$ in \eqref{eq_kernel_function} is an analytic function in the region $(x,y,\tau)\in\C^{3}$ where
\begin{equation}
\Im(\tau)>0 , \quad |\Im(x)| < \pi \Im(\tau) , \quad | \Im(x)| < \Im(y) < 2\pi \Im(\tau) - | \Im(x)|.
\label{eq_analytic_region}
\end{equation}
\end{lemma}
\begin{proof}
From the definition of the $\theta_{1}$-function in \eqref{eq_Jacobi_theta_functions}, it is follows from a straightforward check that the numerator is analytic for integer couplings $(\A,g)$. (Note that the numerator is analytic for any couplings parameters if $x,y$ are real.) It is also known\footnote{Interested readers can check Section XXI of \cite{WhitWat}.} that the $\theta_{1}$-function has simple zeroes at the point $ \pi n + \pi m \tau $, $n,m\in\Z$. The kernel function $\kk(x,y)$ in \eqref{eq_kernel_function} is thus a meromorphic function with poles of order $2g+\A$ at the points
\begin{equation}
x \pm y + 2 \pi n + 2\pi m \tau = 0 \quad (n,m\in\Z)
\label{eq_singular_points}
\end{equation}
which are not included in the region \eqref{eq_analytic_region} above. 
\end{proof}
The integral transform preserves analyticity of the function and the iterated integral representations in \eqref{eq_recursive_integral_solutions} will be an analytic function if the seed function is analytic in the strip $\{z \in\C : |\Im(z)| < \pi \Im(\tau)\}$ (recall that $ \Im(\tau) >0$).

Next we prove that the integral transform is of the form \eqref{eq_ns_lame_eigenfunctions}-\eqref{eq_ns_lame_polynomial_expansion}, and in order to simplify terminology, as we proceed, we define the space $\mathcal{F}_{g}$ as the vector space spanned by the functions 
\begin{equation}
(\theta_{1}(x)^{2})^{\half g} \mathcal{P}(\cos(x))
\end{equation}
where $\mathcal{P}(z)$ an analytic function of  $z$ (in an open domain depending on $|q|$ which includes $[-1,1]$) and $q$ (for $| q | < 1$)  with expansion $\mathcal{P}(z)= \sum_{\ell\in\N_{0}} \mathcal{P}^{(\ell)}(z)q^{2\ell}$, where $\mathcal{P}^{(\ell)}(z)$ is a polynomial with $\text{deg}(\mathcal{P}^{(\ell)}) \leq b + 2\ell$ for some $b>0$.

\begin{lemma}\label{lemma_space_endomorphism}
Let $\psi_{n,g}\in\mathcal{F}_{g}$ then $(K_{g+\A,g}\psi_{n,g})\in\mathcal{F}_{g+\A}$.
\end{lemma}
Most of the remainder of this Section will be devoted to the proof of Lemma \ref{lemma_space_endomorphism}. The normalization constant in \eqref{eq_norm_constant_1} is then explained at the end of this Section.

\begin{proof}[Proof of Lemma \ref{lemma_space_endomorphism}]
We start by re-expressing the integrand as
\begin{equation}\begin{split}
\kk(x,y) \psi_{n,g}(y) &= ( \theta_{1}(x)^{2} )^{\half(g+\A)} \frac{ (\theta_{1}(y)^{2})^{g}}{\Bigl( \theta_{1}(\half( x +y )) \theta_{1}(\half(x-y)) \Bigr)^{2g+\A} } \mathcal{P}_{n,g}(\cos(y)) \\
&=  ( \theta_{1}(x)^{2})^{\half (g+\A)} \e^{\imag \A y} \frac{\Theta_{1}(\e^{2\imag y})^{2 g }}{G^{2(g-\A)} \Theta(\cos(x), \e^{\imag y})^{2g+\A}} \mathcal{P}(\half( \e^{\imag y} + \e^{-\imag y})).
\label{eqproof_integrand_expansion}
\end{split}
\end{equation}
We start by showing that
\begin{equation}
 \frac{\Theta_{1}(\xi^{2})^{2g}}{\Theta(z, \xi)^{2g+\A}} = \sum\limits_{n\in\Z} f_{n}(z) \xi^{n},
\label{eq_Theta_kernel}
\end{equation}
where the sum is absolutely convergent for $| q |^{2} < |\xi| < 1$ (recall that we set $\Im(x)=0$) and that the function $f_{n}(z)$ has the expansion
\begin{equation}
f_{n}(z) = \sum_{\ell=0}^{\infty} f_{n}^{(\ell)}(z) q^{2\ell},
\label{eq_f_n_expansion}
\end{equation}
where $f_{n}^{(\ell)}(z)$ is a polynomial of order $n+2\ell$ for $n+2\ell \geq 0$.

We first consider this in the case $q=0$. Using the definitions in Section \ref{section_notation}, with $q=0$ for \eqref{eq_Theta_kernel}, yields
\begin{equation}
\left. \frac{\Theta_{1}(\xi^{2})^{2g}}{\Theta(z, \xi)^{2g+\A}} \right|_{q=0} = \frac{(1-\xi^{2})^{2g}}{(1-2z\xi + \xi^{2})^{2g+\A}} = \sum_{n \in \N_{0}} f_{n}^{(0)}(z) \xi^{n}.
\end{equation}
Using the binomial series expansion yields
\begin{align}
 \frac{(1-\xi^{2})^{2g}}{(1-2z\xi + \xi^{2})^{2g+\A}}  &= \sum_{n_{1} , n_{2} \in \N_{0}} \binom{2 g}{n_{2}} ( - \xi^{2})^{n_{2}} \binom{-(2g+\A)}{n_{1}} \xi^{n_{1}}( \xi-2z)^{n_{1}} \nonumber \\
&=\sum_{n_{1},n_{2}\in\N_{0}} \sum_{m_{1}=0}^{n_{1}} \binom{2 g}{n_{2}}\binom{-(2g+\A)}{n_{1}} (-1)^{n_{2}} (-2z)^{m_{1}} \xi^{2n_{1} + 2n_{2} - m_{1}} \nonumber \\
&= \sum\limits_{n\in \N_{0}} f_{n}^{(0)}(z) \xi^{n}.
\label{eqproof_binomial_expansion}
\end{align}

Now consider the case $q\neq 0$. It follows from the calculations above and the definitions of the $\Theta_{1}$ and $\Theta$ in Section \ref{section_notation}, that 
\begin{equation}
\frac{\Theta_{1}(\e^{2\imag y})^{2g}}{ \Theta(\cos(x), \e^{\imag y})^{2g+\A}} = \sum_{m \in \N_{0}} f_{m}^{(0)}(z) \xi^{m} \prod\limits_{ k \in \N} \frac{[ ( 1 - q^{2k}\xi^{2})(1-q^{2k}\xi^{-2})]^{2g}}{[(1- 2z q^{2k} \xi  + q^{4k} \xi^{2})(1-2z q^{2k} \xi^{-1} + q^{4k} \xi^{-2})]^{2g+\A}} 
\label{eqproof_lemma_1}
\end{equation}
Expanding all the factors in terms of binomial series yields that the product, on the r.h.s. of \eqref{eqproof_lemma_1}, is given by linear combinations of terms
\begin{equation}
f_{m}^{(0)}(z) \xi^{m} \prod_{k\in\N} ( q^{2k}\xi^{2})^{\mu_{k}^{(1)}}(q^{2k}\xi^{-2})^{\mu_{k}^{(2)}} [ q^{2k} \xi ( q^{2k} \xi - 2 z)]^{\mu_{k}^{(3)}} [ q^{2k}\xi^{-1}(q^{2k}\xi^{-1} - 2 z)]^{\mu_{k}^{(4)}},
\end{equation}
with non-negative integers $m$ and $\mu_{k}^{(l)}$ (l=1,2,3,4). Further expansion yields linear combinations of terms
\begin{equation}
f_{m}^{(0)}(z)\xi^{m} \prod\limits_{k\in\N}  ( q^{2k}\xi^{2})^{\mu_{k}^{(1)}}(q^{2k}\xi^{-2})^{\mu_{k}^{(2)}} (q^{2k}\xi)^{2\mu_{k}^{(3)} - m_{k}^{(3)}} z^{m_{k}^{(3)}} ( q^{2k} \xi^{-1})^{2\mu_{k}^{(4)} - m_{k}^{(4)}}
\label{eqproof_lemma_2}
\end{equation}
with $m_{k}^{(3)} \leq \mu_{k}^{(3)}$ and $m_{k}^{(4)} \leq \mu_{k}^{(4)}$, for all $k\in\N$.
It is clear that \eqref{eqproof_lemma_2} can be written as 
\begin{equation}
 f_{m}^{(0)}(z) z^{N} q^{2L} \xi^{m+M}
\label{eqproof_lemma_3}
\end{equation} with
\begin{equation}
\begin{split}
 M &= \sum_{k\in\N} 2 \mu_{k}^{(1)} - 2\mu_{k}^{(2)} + 2 \mu_{k}^{(3)} - m_{k}^{(3)} - 2\mu_{k}^{(4)} + m_{k}^{(4)},\quad N = \sum_{k\in\N}m_{k}^{(3)}+m_{k}^{(4)} \\
L &= \sum_{k\in\N} k ( \mu_{k}^{(1)} + \mu_{k}^{(2)} + (2\mu_{k}^{(3)} - m_{k}^{(3)}) + ( 2 \mu_{k}^{(4)} - m_{k}^{(4)})),
\end{split}
\end{equation}
where $M\in \Z$ and $N,L\in\N$ (recall that $\mu_{k}^{(l)} \in \N \  \forall k,l \text{ and } 0 \leq m_{k}^{(3)} \leq \mu_{k}^{(3)} ,  \ 0\leq  m_{k}^{(4)} \leq \mu_{k}^{(4)}$).

It follows that the function $f_{n}(z)$ in \eqref{eq_Theta_kernel} is given as linear combinations of the terms in \eqref{eqproof_lemma_3}, with $n = m+ M$, and that it has the expansion in \eqref{eq_f_n_expansion}. We are left to consider the polynomial order: it follows that $m+N \leq n + 2 L$, where $n = m+M$, since 
\begin{equation}
n + 2 L - m - N = M + 2L - N \geq 2 \sum_{k\in\N} ( k + 1)\mu_{k}^{(1)} + (k-1) \mu_{k}^{(2)} + k \mu_{k}^{(3)} +(k-1)\mu_{k}^{(4)} \geq 0
\label{eqproof_lemma_4}
\end{equation}
for non-negative integers $\mu_{k}^{(j)}$ ($j=1,2,3,4$), and $f_{n}^{(\ell)}(z)$, for fixed $n,\ell\in\N_{0}$, is therefore a polynomial of order $n + 2\ell$.

We can then also expand the function $\mathcal{P}_{n,g}(\half( \xi + \xi^{-1}))$ as linear combinations of terms $$ q^{2 \ell'} \xi^{\nu_{1} - 2 \nu_{2}} $$ with non-negative integers $\nu_{1},\nu_{2},\ell'$ satisfying $\nu_{2} \leq \nu_{1} \leq n + 2\ell'$. Since the constant $b$ in Lemma \ref{lemma_space_endomorphism} is arbitrary we can denote the constant by $n$ without loss of generality. The integral transform in \eqref{eq_integral_transform_1} is then given by linear combination of terms
\begin{equation}
(\theta_{1}(x)^{2})^{\half(g+\A)} \int\limits_{-\pi + \imag \varepsilon}^{\pi + \imag \varepsilon} \frac{\d y}{2\pi} \e^{\imag \A y} f_{m}^{(\ell)}(z) q^{2 ( \ell + \ell' )} \e^{\imag (m + \nu_{1} - 2 \nu_{2}) y}
\label{eqproof_lemma_5}\end{equation}
with $m\in \Z, \ell , \ell ' \in \N_{0}$, and $\nu_{1} , \nu_{2}$ satisfying $0 \leq \nu_{2} \leq \nu_{2} \leq n + 2\ell'$ (the interchange of summation and integration is justified since the series is absolutely convergent). By Cauchy's Theorem we obtain that the integral in \eqref{eqproof_lemma_5} is independent of $\varepsilon$, and given by linear combinations of terms $f_{2m- m' - \A}^{(\ell - \ell')}q^{2\ell}$ with non-negative integers satisfying $0 \leq m \leq m' \leq n+2\ell'$ and $0 \leq \ell' \leq \ell$. From the properties of the $f_{n}(z)$ functions it follows that, for fixed $\ell \in \N_{0}$, then the linear combinations of terms $f_{2m- m' - \A}^{(\ell - \ell')}$ is a polynomial of order $n-\A+2\ell$.

\end{proof}

A straightforward check of \eqref{eqproof_binomial_expansion} and \eqref{eq_Gegenbauer_series_rep} also reveals that
\begin{equation}
\int_{\mathcal{C}_{\varepsilon}} \frac{\d y}{2\pi} \frac{(1-\e^{2\imag y})^{2g} \e^{\imag \A y} C_{n}^{(g)}(\cos(y))}{(1 - 2 \cos(x) \e^{\imag y} + \e^{2\imag y})^{2g+\A}} =  2^{-2(g+\A)} \frac{(g)_{n}(2g+\A)_{n-\A}}{n! (g+\A)_{n-\A}} C_{n-\A}^{(g+\A)}(\cos(x)),
\end{equation}
and the constant in \eqref{eq_norm_constant_1} ensures that \eqref{eq_recursive_integral_relation_1} the expansion in \eqref{eq_ns_lame_polynomial_expansion}.

\section{Arbitrary couplings}\label{section_arbitrary_couplings}

For non-integer values of the parameter $\A$, the integral transform in \eqref{eq_integral_transform_1} requires another suitable choice of integration path $\mathcal{C}$. (It is clear from \eqref{eqproof_integrand_expansion} that the more elementary path $\mathcal{C}_{\varepsilon}$ in \eqref{eq_integral_transform_1} only works for integer parameter $\A$.) It has been suggested (see \eg \cite{etingof1994}) that a suitable integration contour would be the Pochhammer loop, in the $\e^{\imag y}$ plane, around the branch points. {Another example of a suitable integration contour can be found in \cite{Langmann_Takemura} for a more general case.} 

In this section we consider the case where the coupling parameter $\A$ can take an arbitrary value and make use of a different integral transform operator. The integral operator used in this Section has the advantage of working for any arbitrary coupling parameter, but the drawback is that the resulting integrals are somewhat more complicated.

\subsection{Results}
We define the integral transform operator $\mathfrak{K}_{g+\A,g}$ as follows: Let $\psi_{n,g}(x)$ be of the form \eqref{eq_ns_lame_eigenfunctions}, with $\mathcal{P}_{n,g}$ in \eqref{eq_ns_lame_polynomial_expansion}. Then
\begin{equation}
(\mathfrak{K}_{g+\A,g}\psi_{n,g})(x) \equiv  q^{-n} \mathcal{M}_{n,g,\A} \int_{-\pi}^{\pi} \frac{\d y}{2\pi} \frac{(\theta_{1}(x)^{2})^{\half( g+\A)} (\theta_{1}(y)^{2})^{\half g}}{( \theta_{4}(\half(x+y))\theta_{4}(\half(x-y)))^{2 g+\A}} \psi_{n,g}(y)
\label{eq_integral_transform_4}
\end{equation}
and normalization constant $\mathcal{M}_{n,g,\A}$, given by
\begin{equation}
\mathcal{M}_{n,g,\A} \equiv  \frac{ (g)_{n} (g+\A)_{n}}{n! ( 2 g + \A)_{n} \mathfrak{h}_{n}} = \frac{{2}^{2g-1}\Gamma ( n+g+1 ) \Gamma( g ) \Gamma( n+ g+\A) \Gamma( 2g + \A) }{\Gamma( g+\A ) \Gamma( n+2 g ) \Gamma( n + 2g+\A) } \label{eq_norm_constant_4}
\end{equation}
with $\mathfrak{h}_{n}$ in \eqref{eq_Gegenbauer_norm_1}, that ensures that the integral transform preserves the expansion in \eqref{eq_ns_lame_polynomial_expansion}.

(Recall the definition of $\theta_{\nu}$-functions in \eqref{eq_Jacobi_theta_functions} and the Lam\ee differential operator in \eqref{eq_Lame_operator}.)

\begin{theorem}\label{theorem_2}
Let $\A \in \R , g\geq 0$, such that $2g+\A>0$, and $n\in\Z$ be fixed. For $\psi_{n,g}$, as in \eqref{eq_ns_lame_eigenfunctions}-\eqref{eq_ns_lame_polynomial_expansion}, 
a solution of \eqref{eq_non_stationary_lame_defining} for parameters $(\A,g)$, with $E = E_{n,g}$ in \eqref{eq_energy_eigenvalues}, then the function
\begin{equation}
( \mathfrak{K}_{g+\A,g} \psi_{n,g})(x),
\label{eq_theorem_integral_transform_4}
\end{equation} 
with $\mathfrak{K}_{g+\A,g}$ in \eqref{eq_integral_transform_4}, is a solution of the non-stationary Lam\ee equation \eqref{eq_non_stationary_lame_defining} for parameters $(\A , g+ \A)$ and $E = E_{n,g+\A}$. Moreover, the function in \eqref{eq_theorem_integral_transform_4} is an analytic function of $x$ and $\tau$ \emph{(} $\Im(\tau)>0$\emph{)} and of the form \eqref{eq_ns_lame_eigenfunctions}-\eqref{eq_ns_lame_polynomial_expansion} for parameters $(\A, g+ \A)$ and with degree $n$.
\end{theorem}
(The proof is given in Section \ref{section_proof_of_theorem_2}.)

It follows from Theorem \ref{theorem_2} that the solutions satisfy the recursive relations
\begin{equation}
\psi_{n,g+\A}(x) = ( \mathfrak{K}_{g+\A,g}\psi_{n,g})(x)
\label{eq_recursive_integral_relation_4}
\end{equation}
which allows us to express the solutions of the non-stationary Lam\ee equation by the recursive procedure explained in Section \eqref{section_sketch_of_solutions}:
\begin{equation}
\begin{split}
\psi_{n,\numb\A + \gn}(x) = (\theta_{1}(x)^{2})^{\half( \numb \A + \gn)} \mathcal{M}_{(\gn)} q^{-\numb n} \int\limits_{[-\pi,\pi]^{\numb}}\!\!\! \frac{\d^{\numb}y}{(2\pi)^{\numb}} \frac{(\theta_{1}(y_{\numb})^{2})^{(\numb-1)\A+\gn}}{(\theta_{4}(\half(x+y_{\numb})) \theta_{4}(\half(x-y_{\numb})))^{(2\numb - 1)\A + 2\gn } } \\
\times \prod\limits_{j=2}^{\numb-1}\Bigl( \frac{(\theta_{1}(y_{j})^{2})^{(j-1)\A + \gn} }{( \theta_{4}( \half( y_{j+1} - y_{j} )) \theta_{4}(\half( y_{j+1} - y_{j})) )^{(2j-1)\A +2\gn} }\Bigr)  \frac{ (\theta_{1}(y_{1})^{2})^{\half\gn} \psi_{n,\gn}(y_{1})}{(\theta_{4}(\half(y_{2} + y_{1} )) \theta_{4}(\half(y_{2}- y_{1})) )^{2\gn+\A}} ,
\end{split}\end{equation}
with $\mathcal{M}_{(\gn)}\equiv \mathcal{M}_{n,(\numb-1)\A+\gn,\A}\cdots\mathcal{M}_{n, \A+ \gn,\A}$. In the special cases where $\gn = \pn \in\{ 0 , 1 \}$ we can use the seed functions in \eqref{eq_seed_functions_1} and the recursive procedure yields the explicit integral representations of the solutions:
\begin{equation}
\begin{split}
\psi_{n,\numb\A}(x) = (\theta_{1}(x)^{2})^{\half  \numb \A} \mathcal{M}_{(0)} q^{-\numb n} \int\limits_{[-\pi,\pi]^{\numb}}\!\!\! \frac{\d^{\numb}y}{(2\pi)^{\numb}} \frac{(\theta_{1}(y_{\numb})^{2})^{(\numb-1)\A}}{(\theta_{4}(\half(x+y_{\numb})) \theta_{4}(\half(x-y_{\numb})))^{(2\numb - 1)\A} } \\
\times \prod\limits_{j=1}^{\numb-1}\Bigl( \frac{(\theta_{1}(y_{j})^{2})^{(j-1)\A } }{( \theta_{4}( \half( y_{j+1} - y_{j} )) \theta_{4}(\half( y_{j+1} - y_{j})) )^{(2j-1)\A } }\Bigr)  \cos(n y_{1}) ,
\end{split}\end{equation}
and
\begin{equation}
\begin{split}
\psi_{n,\numb\A + 1}(x) = (\theta_{1}(x)^{2})^{\half( \numb \A + 1)} \mathcal{M}_{(1)} q^{-\numb n} \int\limits_{[-\pi,\pi]^{\numb}}\!\!\! \frac{\d^{\numb}y}{(2\pi)^{\numb}} \frac{(\theta_{1}(y_{\numb})^{2})^{(\numb-1)\A+1}}{(\theta_{4}(\half(x+y_{\numb})) \theta_{4}(\half(x-y_{\numb})))^{(2\numb - 1)\A + 2 } } \\
\times \prod\limits_{j=1}^{\numb-1}\Bigl( \frac{(\theta_{1}(y_{j})^{2})^{(j-1)\A + 1} }{( \theta_{4}( \half( y_{j+1} - y_{j} )) \theta_{4}(\half( y_{j+1} - y_{j})) )^{(2j-1)\A +21} }\Bigr) \Bigl( \frac{\sin((n+1)y_{1})^{2}}{G^{3/\A}\theta_{1}(y_{1})^{2}}\Bigr)^{\half}.
\end{split}\end{equation}

\subsection{Proof of Theorem \ref{theorem_2}}\label{section_proof_of_theorem_2}
The proof of Theorem \ref{theorem_2} will follow the same steps as the proof of Theorem \ref{theorem_1} in Section \ref{section_proof_of_theorem_1}: we show that the integral transform yields an analytic (normalizable) function (see Lemma \ref{lemma_normalizable}), show that the integral kernel satisfies the generalized kernel function identity (see Lemma \ref{lemma_gen_KF_id_4}), and that the integral transform preserves the form in \eqref{eq_ns_lame_eigenfunctions}-\eqref{eq_ns_lame_polynomial_expansion}. The final part will first be shown for the special case where $\A \notin \Z$ as the proofs are less technical and also allows us to discuss the uniqueness of the solutions.

\begin{lemma}\label{lemma_normalizable}
Let $\psi_{n,g}(x)$ be of the form \eqref{eq_ns_lame_eigenfunctions}-\eqref{eq_ns_lame_polynomial_expansion} and assume that $\psi_{n,g}$ satisfies \eqref{eq_recursive_integral_relation_4}. Then there exists a finite constant $C$, depending on $|q|$ and the parameters $(\A, g)$, such that
\begin{equation}
\lVert \psi_{n,g+\A} \rVert_{L^{2}}^{2} \leq C^{2} \lVert \psi_{n,g} \rVert_{L^{2}}^{2}
\label{eq_lemma_inequality}
\end{equation}
with respect to the standard $L^{2}([-\pi , \pi ] ,\dd x)$ norm.
\end{lemma}
\begin{proof}
From \eqref{eq_recursive_integral_relation_4} and \eqref{eq_integral_transform_4}, we obtain that
\begin{align}
\lVert \psi_{n,g+\A}(x) \rVert^{2} &=  \lVert (\mathfrak{K}_{g+\A,g} \psi_{n,g}) \rVert^{2} \leq \int\limits_{-\pi}^{\pi}\frac{\d x}{2\pi} \int\limits_{-\pi}^{\pi} \frac{\d y}{2\pi} \Bigl| q^{-n} \frac{(\theta_{1}(x)^{2})^{\half(g+\A)} (\theta_{1}(y)^{2})^{\half g}}{(\theta_{4}(\half(x+y))\theta_{4}(\half(x-y)))^{2g+\A}} \Bigr|^{2} \left| \psi_{n,g}(y) \right|^{2} \nonumber \\
& \leq \lVert \psi_{n,g} \rVert^{2} \int\limits_{-\pi}^{\pi}\frac{\d x}{2\pi} \int\limits_{-\pi}^{\pi} \frac{\d y}{2\pi} \Bigl| q^{-n} \frac{(\theta_{1}(x)^{2})^{\half(g+\A)} (\theta_{1}(y)^{2})^{\half g}}{(\theta_{4}(\half(x+y))\theta_{4}(\half(x-y)))^{2g+\A}} \Bigr|^{2},
\label{eqproof_lemma_inequality}
\end{align}
by Cauchy-Schwarz inequality. We note that the r.h.s. of \eqref{eqproof_lemma_inequality} can be written as  $$\lVert \psi_{n,g}\rVert^{2} \lVert \mathfrak{K}_{g+\A,g}\rVert^{2}_{L^{2}([-\pi , \pi]^{2} , \d x \d y )}.$$ From the definitions in \eqref{eq_Jacobi_theta_functions} it is clear\footnote{Interested readers can check \eg \cite{WhitWat}.} that $\theta_{1}(x)$ is a bounded analytic function and that $\theta_{4}(x)$ is a non-zero (bounded) analytic function, for $| q | < 1$ and $x$ in the strip $ \{ x \in \C : |\Im(x) |< \pi \Im(\tau)\}$. Therefore, there exists a finite constant $c >0$, depending on $|q|$ and the model parameters, such that
\begin{equation}
\left| \frac{(\theta_{1}(x)^{2})^{\half (g+\A)} (\theta_{1}(y)^{2})^{\half g}}{(\theta_{4}(\half(x+y))\theta_{4}(\half(x-y)))^{2g+\A}} \right| \leq c.
\end{equation}
Taking $c^{2} = |q|^{n}  C^{2}$ yields \eqref{eq_lemma_inequality}.
\end{proof}
It follows that the integral transform also preserves analyticity and that the iterative integral representation is analytic if the seed function is analytic. (see also Section \ref{section_conclusions} for how to transform the irregular solutions.) This is clearly the case for the explicit integral representations with seed functions $\psi_{n,\pn}$ in \eqref{eq_seed_functions_1}.

(Recall the Lam\ee differential operator $H(x;g)$ in \eqref{eq_Lame_operator} and $\theta_{\nu}$-functions in \eqref{eq_Jacobi_theta_functions}.)
\begin{lemma}\label{lemma_gen_KF_id_4}
For $g,\A$ complex constants, $x,y$ complex variables, and the kernel function
\begin{equation}
k_{\mu}(x,y;g,\A) \equiv \frac{(\theta_{1}(x)^{2})^{\half (g+\A)} (\theta_{1}(y)^{2})^{\half g}}{(\theta_{\mu}(\half(x+y)) \theta_{\mu}(\half(x-y)))^{2g+\A}}
\label{eq_kernel_function_4}
\end{equation}
with $\mu \in \{1,2,3,4\}$ fixed. Then 
\begin{equation}
\Bigl(\frac{2\imag }{\pi} \A \frac{\partial}{\partial \tau} + H(x;g+\A) - H(y;g) - C_{\mu} \Bigr) k_{\mu}(x,y)=0
\label{eq_kernel_function_identity_mu}
\end{equation}
with $H$ in \eqref{eq_Lame_operator} and
\begin{equation}
C_{\mu} = \begin{cases}
4\A(1-2g-\A) \frac{\eta_{1}}{\pi} + 6 \A( 2g + \A) (\frac{\eta_{1}}{\pi} - \frac{1}{12} ),  & \text{ if } \mu=1,2 \\
\A(2g+\A) + 4\A(1-2g-\A) \frac{\eta_{1}}{\pi} + 6 \A( 2g + \A) (\frac{\eta_{1}}{\pi} - \frac{1}{12} ), &\text{ if } \mu=3,4
\end{cases}
\label{eq_kernel_function_constants}
\end{equation}
\end{lemma} \noindent
(The proof is given in Appendix \ref{appendix_proof_of_lemmas}.)

We see that the integral transform in \eqref{eq_integral_transform_4} has \eqref{eq_kernel_function_4} with $\mu=4$ as integral kernel (up to multiplication by the $q$ factor). The action of the operator $(2\imag / \pi) \A (\partial / \partial \tau) +H $ follows from the generalized kernel function identity:
\begin{align*}
&\Bigl( \frac{2\imag}{\pi} \A \frac{\partial}{\partial \tau} + H(x ; g+\A)\Bigr) ( \mathfrak{K}_{g+\A,g}\psi_{n,g})(x) \\
&= q^{-n} \int\limits_{-\pi}^{\pi} \frac{\d y}{2\pi} \psi_{n,g}(y) \Bigl( \frac{2\imag}{\pi} \A \frac{\partial}{\partial \tau} + H(x ; g+\A)\Bigr) k_{4}(x,y)  \\
&+  \int\limits_{-\pi}^{\pi} \frac{\d y}{2\pi}\Bigl( q^{-n} \frac{2\imag}{\pi} \A \frac{\partial \psi_{n,g}(y)}{\partial \tau} k_{4}(x,y)  + \frac{2\imag }{\pi} \A \frac{\partial q^{-n}}{\partial \tau} k_{4}(x,y) \psi_{n,g}(y) \Bigr) \\
&= q^{-n} \int\limits_{-\pi}^{\pi} \frac{\d y}{2\pi} \Bigl( \psi_{n,g}(y) \Bigl( H(y ; g) + 2\A n + C_{4} \Bigr) k_{4}(x,y) + \frac{2\imag}{\pi} \A \frac{\partial \psi_{n,g}(y)}{\partial \tau} k_{4}(x,y) \Bigr)\\ 
&= B.T. +  q^{-n} \int\limits_{-\pi}^{\pi} \frac{\d y}{2\pi} k_{4}(x,y)   \Bigl( \frac{2\imag}{\pi} \A \frac{\partial}{\partial \tau} + H(y ; g) +2 \A n + C_{4}\Bigr) \psi_{n,g}(y) 
\end{align*}
where $B.T.$ is given by
\begin{equation}
B.T. = q^{-n} \int\limits_{-\pi}^{\pi} \frac{\d y}{2\pi} \frac{\partial}{\partial y} \Bigl( \psi_{n,g}(y) \frac{\partial}{\partial y} k_{4}(x,y) - k_{4}(x,y) \frac{\partial}{\partial y} \psi_{n,g}(y) \Bigr) = 0.
\end{equation}

Note that the constant $C_{4}$, and also $C_{3}$, can be written as 
\begin{equation}\begin{split}
C_{4} &= \bigl( (g+\A)^{2} - g^{2}\bigr) - \bigl( 4 (g+\A) ( g+\A-1)  - 4 g ( g- 1) \bigr) \frac{\eta_{1}}{\pi} + \bigl( 6 (g+\A)^{2} - 6 g^{2} \bigr) \Bigl( \frac{\eta_{1}}{\pi} - \frac{1}{12}\Bigr)  \\  &= E_{0,g+\A} - E_{0,g},
\end{split}
\end{equation}
with $E_{n,g}$ in \eqref{eq_energy_eigenvalues}. It then is a straightforward check that the function in \eqref{eq_integral_transform_4} is a solutions of \eqref{eq_non_stationary_lame_defining} with $E=E_{n,g+\A}$:
\begin{equation}
\begin{split}
E&=E_{n,g} + 2\A n +C_{4} \\
& = \bigl( (n+g)^{2} -4g(g-1) \frac{\eta_{1}}{\pi} + 6g^{2}\Bigl( \frac{\eta_{1}}{\pi}-\frac{1}{12}\Bigr) \bigr) + 2\A n \\ &+ \Bigl( \A(2g+\A) + 4\A(1-2g-\A) \frac{\eta_{1}}{\pi} + 6 \A( 2g + \A) (\frac{\eta_{1}}{\pi} - \frac{1}{12} )\Bigr) \\
&= E_{n,g+\A}.
\end{split}
\end{equation}

Before proceeding with Lemma \ref{lemma_polynomial_expansion}, we wish to discuss the special cases where the model parameter $\A$ is not an integer. The non-integer case pertains to the issue of uniqueness of the solutions of \eqref{eq_non_stationary_lame_defining}, and the uniqueness of the functions $\psi_{n,g}$ defined by \eqref{eq_recursive_integral_relation_4}. The proof of the last part of Theorem \ref{theorem_2} can be done with less technical, and more accessible, arguments when $\A$ is not an integer. We therefore also give the proof for this special case, for the convenience of the reader. 

First we show that the integral transform in \eqref{eq_integral_transform_4} indeed reduces to the Gegenbauer polynomials in the trigonometric limit.
Consider the $q \downarrow 0$ limit of the equation
\begin{equation}
\Bigl( -2 \A q \frac{\partial}{\partial q} - H(x;g+\A) -E_{n,g+\A} \Bigr) ( \mathfrak{K}_{g+\A,g} \psi_{n,g})(x) = 0.
\label{eqproof_lemma_6}
\end{equation}
Using the definition of the integral operator in \eqref{eq_integral_transform_4}, $E_{n,g}$ in \eqref{eq_energy_eigenvalues}, and the trigonometric limits in \eqref{eq_trig_limit_Weierstrass}-\eqref{eq_trig_limit_theta} yields
\begin{equation}
\Bigl( -\frac{\partial^{2}}{\partial x^{2}} + (g+\A)(g+\A-1)\frac{1}{\sin(x)^{2}} - (n+g+\A )^{2} \Bigr) (\sin(x)^{2})^{\half (g+\A)} \lim_{q\downarrow 0} \int\limits_{-\pi}^{\pi}\frac{\d y}{2\pi} q^{-n} k_{4}(x,y) \psi_{n,g}(y) = 0.
\label{eqproof_lemma_7}
\end{equation}
It then follows that either 
\begin{equation}
\lim_{q\downarrow 0} \int\limits_{-\pi}^{\pi}\frac{\d y}{2\pi} q^{-n} k_{4}(x,y) \psi_{n,g}(y) = 0
\label{eqproof_lemma_8}
\end{equation} or that
\begin{equation}
\lim_{q\downarrow 0} \int\limits_{-\pi}^{\pi}\frac{\d y}{2\pi} q^{-n} k_{4}(x,y) \psi_{n,g}(y) = \mathcal{M} C_{n}^{(g)}(\cos(x)),
\end{equation}
for some non-zero constant $\mathcal{M}$. (The normalization constant is determined at the end of this Section.)
Suppose that \eqref{eqproof_lemma_8} holds; it was shown in \cite{ruijsenaars2009} that if $\phi \in L^{1}([-\pi , \pi ] , \d x )$, then 
\begin{equation}
\int_{-\pi}^{\pi} \frac{\d y}{2\pi} \frac{(\theta_{1}(y)^{2})^{g}}{ \Bigl( \theta_{4}(\half(x-y) ) \theta_{4}(\half(x+y))\Bigr)^{2g+\A}} \phi(y) = 0 \quad \text{ if and only if } \phi(y) = 0.
\end{equation}(We wish to point out that the proof in \cite{ruijsenaars2009} was for $\A=\Im(\tau)=0$ but that the arguments still hold, as long as $2g+\A>0$.) Therefore there must exist an integer $m$ such that
\begin{equation}
f(\cos(x)) \equiv \lim_{q\downarrow 0} q^{m} \int\limits_{-\pi}^{\pi}\frac{\d y}{2\pi} q^{-n} k_{4}(x,y) \psi_{n,g}(y)  \neq 0.
\label{eqproof_lemma_9}
\end{equation}
It then follows that $f(\cos(x))$ satisfies
\begin{equation}
\Bigl( -\frac{\partial^{2}}{\partial x^{2}} + (g+\A)(g+\A-1)\frac{1}{\sin(x)^{2}} - (n+g+\A )^{2} - 2\A m \Bigr) (\sin(x)^{2})^{\half (g+\A)} f(\cos(x)) = 0,
\label{eqproof_lemma_14}
\end{equation}
but the eigenvalues of the Schr{\"o}dinger operator, with the P{\"o}schl-Teller potential, does not match, \ie there is no integer $k$ such that $(k+g+\A)^{2} = (n+g+\A )^{2} + 2\A m$ for $\A\in\R\setminus\Z$. Then the only (normalizable) solution of  \eqref{eqproof_lemma_14} is $f(\cos(x))=0$ and we have a contradiction.

We turn to show that the polynomials $\mathcal{P}^{(\ell)}_{n,g}(z)$ are of order $n+2\ell$ and present a method for computing the polynomials explicitly:
Let us write the solutions $\psi_{n,g}$ of \eqref{eq_non_stationary_lame_defining} as 
\begin{equation}
\psi_{n,g}(x) = (\sin(x)^{2})^{\half g} \tilde{\mathcal{P}}_{n,g}(\cos(x)), \quad \tilde{\mathcal{P}}_{n,g}(z) \equiv  \prod_{k\in\N}( 1 - 2 q^{2k} ( 2 z^{2} - 1) q^{4k})^{2g} \mathcal{P}_{n,g}(z)
\label{eqproof_lemma_10}
\end{equation}
and it is clear that the functions $\tilde{\mathcal{P}}_{n,g}(z)$ has the expansion
\begin{equation}
\tilde{\mathcal{P}}_{n,g}(z) = \sum_{\ell \in\N_{0}} \tilde{\mathcal{P}}^{(\ell)}_{n,g}(z) q^{2 \ell}.
\end{equation}
If $\psi_{n,g}$ is a solution of the non-stationary Lam\ee equation, then the polynomials $\tilde{\mathcal{P}}^{(\ell)}_{n,g}(z)$, for a fixed $\ell \in\N_{0}$, have the expansion
\begin{equation}
\tilde{\mathcal{P}}^{(\ell)}_{n,g}(z) = \sum_{m=-\ell}^{\ell} d^{(\ell)}_{n,n+2m} C_{n+2m}^{(g)}(z) \quad ( \ell \in \N_{0} )
\label{eq_ns_lame_polynomial_Gegenbauer_expansion}
\end{equation}
for some constants $d^{(\ell)}_{n,m}$ depending on the model parameters. Since the functions $\mathcal{P}$ and $\tilde{\mathcal{P}}$ are related by \eqref{eqproof_lemma_10}, it follows that $\tilde{\mathcal{P}}_{n,g}^{(\ell)}(z)$ is a polynomial of order $n+2\ell$. We proceed with the following calculations to show this.  We can expand the functions $\tilde{\mathcal{P}}^{(\ell)}_{n,g}(z)$ in terms of the Gegenbauer polynomials, \ie
\begin{equation}
\tilde{\mathcal{P}}^{(\ell)}_{n,g}(z) = \sum_{m\in\N_{0}} d^{(\ell)}_{n,m} C_{m}^{(g)}(z) , \quad d^{(\ell)}_{n,m}\equiv \frac{\int_{-1}^{1} \d z ( 1 - z^{2} )^{g - \half} \tilde{\mathcal{P}}^{(\ell)}_{n,g}(z)C_{m}^{(g)}(z)}{\int_{-1}^{1} \d z ( 1 - z^{2} )^{g - \half} C_{m}^{(g)}(z)^{2}} = \frac{1}{\lVert C_{m}^{(g)} \rVert^{2}} \langle C_{m}^{(g)}(z) ,  \tilde{\mathcal{P}}^{(\ell)}_{n,g}(z) \rangle
\end{equation}
with $\langle \cdot , \cdot \rangle$ the standard inner product of $L^{2}([-1,1], (1-z^{2})^{g-\half} \d z )$ divided by $\pi$. We now show that $d^{(\ell)}_{n,m}=0$, if $| n - m | > 2\ell$, using \eqref{eqproof_lemma_11}: define the differential operator
\begin{equation}
h^{(0)}(\cos(x) ; g) \equiv (\sin(x)^{2})^{-\half g } H^{(0)}(x;g) (\sin(x)^{2})^{\half g} = - \frac{\partial^{2}}{\partial x^{2}} - 2g \cot(x) \frac{\partial}{\partial x} + g^{2}
\label{eq_Gegenbauer_diff_op}
\end{equation}
which is symmetric with respect to the inner product $\langle \cdot , \cdot \rangle$ by definition. Then consider
\begin{align}
\langle h^{(0)}(z;g) C_{m}^{(g)}(z) , \tilde{\mathcal{P}}_{n,g}^{(\ell)}(z) \rangle = \langle  C_{m}^{(g)}(z) , h^{(0)}(z;g)\tilde{\mathcal{P}}_{n,g}^{(\ell)}(z) \rangle. \label{eqproof_lemma_17}\end{align}
The action of the operator $h^{(0)}(x;g) $ in \eqref{eq_Gegenbauer_diff_op} on the polynomials $\tilde{\mathcal{P}}_{n,g}^{(\ell)}$ is obtained by inserting the ansatz \eqref{eqproof_lemma_10} in \eqref{eq_non_stationary_lame_defining}. Straightforward calculations yield
\begin{multline}
\sum_{\ell\in\N_{0}} \Bigl[ \bigl( -4 \A \ell + H^{(0)}(x;g) -E_{n,g}^{(0)} \bigr) (\sin(x)^{2})^{\half g} \tilde{\mathcal{P}}^{(\ell)}_{n,g}(z) q^{2 \ell} \\  -  \sum_{ \ell' \in \N } \Bigl( V_{\ell'}(\cos(x))  + E_{n,g}^{(\ell')} \Bigr)(\sin(x)^{2})^{\half g} \tilde{\mathcal{P}}^{(\ell)}_{n,g}(z) q^{2 (\ell + \ell')} \Bigr] = 0,
\label{eqproof_lemma_15}
\end{multline}
with $H^{(0)}(x,g)$ the Schr{\"o}dinger operator in \eqref{eq_Poschl_Teller_Hamiltonian}, $E_{n,g} = \sum_{\ell\in\N_{0}}E_{n,g}^{(\ell)} q^{2\ell}$, and $V_{\ell}(\cos(x))$ a short-hand for the $q$ dependent, non-constant terms of the $\wp$-function in \eqref{eq_non_stationary_lame_defining} (see also \eqref{eq_Weierstrass_P_expansion}), \ie 
\begin{equation}
-8g(g-1)\sum_{\ell\in\N} \frac{q^{2\ell}}{1- q^{2\ell}}\cos(2\ell y) =  \sum_{\ell\in\N} V_{\ell}(\cos(x))q^{2\ell}.\label{eqproof_lemma_19} \end{equation} (It is clear, from the relation above, that $V_{\ell}(\cos(x))$ is a polynomial of degree $\ell$ in $\cos(x)^{2}$).
Collecting the same powers of $q$ in \eqref{eqproof_lemma_15} yields the relations
\begin{multline}
\bigl( - 4 \A \ell + H^{(0)}(x;g) -E_{n,g}^{(0)} \bigr) (\sin(x)^{2})^{\half g} \tilde{\mathcal{P}}^{(\ell)}_{n,g}(\cos(x))  \\ = \sum_{\ell'=1}^{\ell} \Bigl( V_{\ell'}(\cos(x)) + E_{n,g}^{(\ell')} \Bigr)  (\sin(x)^{2})^{\half g} \tilde{\mathcal{P}}^{(\ell- \ell')}_{n,g}(\cos(x))
\label{eqproof_lemma_11}
\end{multline}
for all $\ell \in \N_{0}$.
Inserting \eqref{eqproof_lemma_11} into \eqref{eqproof_lemma_17} yields
\begin{equation}
(E_{m,g}^{(0)} - E_{n,g}^{(0)} - 4\A\ell) d^{(\ell)}_{n,m} = \sum\limits_{\ell' = 1}^{\ell} \langle C_{m}^{(g)}(z), (V_{\ell'} + E_{n,g}^{(\ell')}) \tilde{\mathcal{P}}^{(\ell - \ell')}_{n,g}(z) \rangle
\label{eqproof_lemma_18}
\end{equation}
Then for any $\ell \in\N $ insert the relation 
\begin{align}
\sum_{\ell'=1}^{\ell} \langle C_{m}^{(g)}(z) , V_{\ell'}(z) \tilde{\mathcal{P}}_{n,g}^{(\ell - \ell')}(z)\rangle &= \sum_{\ell' = 1 }^{\ell} \sum_{\nu=0}^{\ell'} v_{\nu} \langle C_{m}^{(g)}(z) , z^{2\nu} \mathcal{P}_{n,g}^{(\ell - \ell') }(z)\rangle \\
&= \sum_{\ell'=1}^{\ell} \sum_{\nu=0}^{\ell} \sum_{\mu=-(\ell - \ell')}^{(\ell - \ell')}v_{\nu} d_{n,\mu}^{(\ell-\ell')} \langle  C_{m}^{(g)}(z) , z^{2\nu}C_{n+2\mu}^{(g)}(z)\rangle,
\label{eqproof_lemma_12}
\end{align}
for some coefficients $v_{\nu}$ obtained from the expansion of $V_{\ell'}(z)$ in \eqref{eqproof_lemma_19}. Using the three-term recurrence relation in \eqref{eq_three_term_rec} and the orthogonality of the Gegenbauer polynomials yields\footnote{The term $ z^{2\nu}C_{n+2\mu}^{(g)}(z) $ in \eqref{eqproof_lemma_12}, is given by linear combinations of Gegenbauer polynomials $C_{s}^{(g)}(z)$ with $s= n +2\mu - 2\nu , n + 2 \mu - 2\nu +2 , \ldots , n + 2\mu + 2\nu.$} that the l.h.s. of \eqref{eqproof_lemma_18} is identically $0$ if $m > n + 2 (\ell - \ell') + 2 \ell' = n + 2\ell $ or if  $m < n - 2( \ell - \ell') - 2\ell' = n - 2 \ell$. 

\begin{remark}\label{remark_resonance_cases}
It is clear that the arguments given above do not hold in the special cases where the generalized energy differences vanish, \ie $$ E_{m,g}^{(0)} - E_{n,g}^{(0)}-4\A \ell = (m-n)(n+m+2g) - 4\A \ell = 0 \quad (n,m\in\N_{0},\ell\in\N),$$ which we refer to as resonance cases \cite{langmann:7,AtaiLangmann2016_2}. (Recall that $E_{n,g}^{(0)} = (n+g)^{2}$.) In those cases the coefficients $d_{n,m}^{(\ell)}$ can be given arbitrary value and we can set $d_{n,m}^{(\ell)} = 0$, if $| n- m| \geq 2\ell$, consistently and without loss of generality.

These resonance cases cannot occur if the model parameters $(\A,g)$ (recall from Section \ref{section_sketch_of_solutions}) are sufficiently irrational or if they have a non-zero imaginary part. (The imaginary $\A$ case is also discussed more extensively in \cite{AtaiLangmann2016_2}.) The resonance cases seem to be most prevalent if the model parameters $(\A, g )$ take integer values, yet it was proven in Section \ref{section_integer couplings} that similar results\footnote{Meaning that the $\mathcal{P}^{(\ell)}_{n,g}(z)$ are also polynomials of order $n+2\ell$ and not that they have the expansion in \eqref{eq_ns_lame_polynomial_Gegenbauer_expansion}.} also hold in the integer cases. This yields strong indications that the resonance singularities are removable, as conjectured in \cite{AtaiLangmann2016_2}. It was also proven in \cite{Takemura2004_1} that the resonances do not occur in the $\A=0$ case, for $g \notin -\N / 2$.
\end{remark}

\begin{remark}\label{remark_uniqueness}
If there are no resonances for all $\ell \in\N_{0}$ and $m=n-2\ell , \ldots , n+2\ell$, where $n\in\N_{0}, \A\in\R, g\geq 0$ are fixed, then \eqref{eqproof_lemma_10}, determined by \eqref{eq_ns_lame_polynomial_Gegenbauer_expansion}, \eqref{eqproof_lemma_11}, and the normalization in \eqref{eq_ns_lame_polynomial_expansion}, yield a solution of the non-stationary Lam\ee equation \eqref{eq_non_stationary_lame_defining} which is unique, up to multiplication by an analytic function of $q^{2}$ that preserves \eqref{eq_ns_lame_polynomial_expansion}. We also believe that this perturbative construction will yield an analytic function: It has been proven by Takemura \cite{Takemura2004_1} for the Heun case, \ie $\A=0$, that these coefficients converge and that the perturbative construction yields an $L^{2}$ function. A special case of the approach for the elliptic Calogero-Sutherland model by Langmann \cite{langmann:7} also proved the convergence using an unconventional basis for the perturbative expansion (see discussion in Section 1 of \cite{AtaiLangmann2016_2}).

We conjecture that the integral transforms also yield unique solutions (up to multiplication) in the cases where no resonances can occur for distinct seed functions.
\end{remark}

\begin{lemma}\label{lemma_polynomial_expansion}
Let $\psi_{n,g}$ be as in Theorem \ref{theorem_2}, then Eq. \eqref{eq_theorem_integral_transform_4} is of the form
\begin{equation}
(\theta_{1}(x)^{2})^{\half ( g + \A)} \Bigl( C_{n}^{(g+\A)}(\cos(x)) + \sum_{\ell\in\N} \mathcal{P}_{n,g}^{(\ell)}(\cos(x)) q^{2\ell}\Bigr),
\end{equation}
with $C_{n}^{(g)}(z)$ the Gegenbauer polynomials and $\mathcal{P}_{n,g}^{(\ell)}(z)$ a polynomial of order $n+2\ell$.
\end{lemma}
\noindent(The proof is given at the end of this Section.)

Before we proceed with the proof, we wish to recall some well-known facts about the Gegenbauer polynomials: Let $p_{k}$ denote an arbitrary polynomial of degree $k\in\N_{0}$, then
\begin{equation}
\int_{-\pi}^{\pi} \frac{\d y}{2\pi} (\sin(y)^{2})^{g} C_{n}^{(g)}(\cos(y)) p_{k}(\cos(y)) = 0 \quad \forall \ k < n .
\label{eq_Gegenbauer_orthog}
\end{equation}

If $p_{n}(z)$ satisfies the Gegenbauer differential equation (see \eqref{eq_Gegenbauer_diff_equation}) and has the expansion
\begin{equation}
p_{n}(z) = \frac{2^{n} ( g )_{n}}{n!} z^{n}  +\mathcal{O}(z^{n-1}),
\end{equation}
then $p_{n}(z)=C_{n}^{(g)}(z)$.

\begin{proof}[Proof of Lemma \ref{lemma_polynomial_expansion}.]
Let us start by re-expressing the kernel function as
\begin{multline}
\frac{1}{(\theta_{4}(\half(x+y)) \theta_{4}(\half (x-y)))^{2g+\A}} \\ = \exp\left\{(2g+\A) \sum_{m\in\N} \frac{(-1)^{m}}{m}\sum_{\mu=0}^{m} \binom{m}{\mu} \frac{(-2)^{\mu}q^{2m-\mu}}{1-q^{2(2m-\mu)}}\Bigl( \cos( x +y )^{\mu} + \cos(x-y)^{\mu}\Bigr)\right\}.
\end{multline}
This relation follows from the simple expansion
\begin{align}
\frac{1}{\theta_{4}(x)^{\lambda}}= \prod_{k\in\N}\e^{-\lambda \ln(1-2q^{2k-1}\cos(2x)+q^{4k-2})}
& = \prod_{k\in\N}\e^{\lambda \sum_{m\in\N} \frac{(-1)^{m}}{m} \sum_{\mu=0}^{m}\binom{m}{\mu} q^{(2k-1)(2m-\mu)} (-2\cos(2x))^{\mu}} \\& =  \e^{\lambda \sum_{m\in\N} \frac{(-1)^{m}}{m} \sum_{\mu=0}^{m}\binom{m}{\mu} \frac{(-2)^{\mu}q^{2m-\mu}}{1-q^{2(2m-\mu)}}  (\cos(2x))^{\mu} }.
\end{align}
This allows us to write the integral transform in \eqref{eq_integral_transform_4} as
\begin{multline}
q^{-n}\mathcal{M}_{n,g,\A} (\theta_{1}(x)^{2})^{\half(g+\A)} \int\limits_{-\pi}^{\phantom{-}\pi} \frac{\d y}{2\pi} (\sin(y)^{2})^{g} \Bigl( C_{n}^{(g)}(\cos(y)) + \sum_{\ell\in\N} \sum_{m=-\ell}^{\ell} d_{n,n+2m}^{(\ell)} C_{n+2m}^{(g)}(\cos(y)) q^{2\ell} \Bigr) \\ 
\times \prod_{k\in\N}( 1 - 2q^{2k}( 2 \cos(y)^{2}- 1) + q^{4k})^{2g} \\ \times  \exp\left((2g+\A) \sum_{m\in\N} \frac{(-1)^{m}}{m}\sum_{\mu=0}^{m} \binom{m}{\mu} \frac{(-2)^{\mu}q^{2m-\mu}}{1-q^{2(2m-\mu)}}\Bigl( \cos( x +y )^{\mu} + \cos(x-y)^{\mu}\Bigr)\right).
\label{eqproof_lemma_16}
\end{multline}
The Taylor series of the second row of \eqref{eqproof_lemma_16} shows that every power of cosine is multiplied with, at least, the same power of $q$. Conversely, for a fixed power $m$ (say) of $q$, we do not get a higher power of cosine than $m$. It follows therefore that the coefficients in front of $q^{n+2\ell}$ is a polynomial of order $n+2\ell$ in $\cos(y)$, $\cos(x)$, and $\cos(x)\cos(x)$.
From \eqref{eq_Gegenbauer_orthog} it follows that the first non-zero term after integration has at least a power of $q^{n}$.  We therefore have a polynomial of order $n+2\ell$, in $\cos(x)$, in front of every $q^{2\ell}$ term, after integration.
To see that we do not get any powers $n+2\ell+1$ we recall that the Gegenbauer polynomials are even (resp. odd) polynomials for $n$ even (resp. odd). Then the integral vanished identically for odd (resp. even) polynomials of $\cos(y)$, if the Gegenbauer polynomial is even (resp. odd), \ie
\begin{equation}
\int\limits	_{-\pi}^{\pi} \frac{\d y}{2\pi} (\sin(y)^{2})^{g} C_{2n}^{(g)}(\cos(y)) \cos(y)^{2k+1} = \int\limits_{-\pi}^{\pi} \frac{\d y}{2\pi} (\sin(y)^{2})^{g} C_{2n+1}^{(g)}(\cos(y)) \cos(y)^{2k} =0 \quad \forall n,k\in\N_{0}.
\end{equation}

We now turn to calculating the normalization constant. Let us also consider the term in the integral which has the highest power of $\cos(x)$, with the lowest power of $q$, that will not vanish after integration. The term is given by
\begin{align}
\mathcal{P}_{n,g+\A}^{(0)}(\cos(x)) = \mathcal{M}_{n,g,\A} \int\limits_{-\pi}^{\phantom{-}\pi} \frac{\d y}{2\pi} (\sin(y)^{2})^{g} C_{n}^{(g)}(\cos(y))  \binom{-(2g+\A)}{n} (-2)^{n} \cos(x)^{n} \cos(y)^{n} + \ldots 
\end{align}
where the ''$\ldots$'' stand for terms that will not contribute. Using $\cos(y)^{n} = (n! / 2^{n}(g)_{n}) C_{n}^{(g)}(\cos(y)) + \text{ l.d. }$ yields that
\begin{equation}
\mathcal{P}_{n,g+\A}^{(0)}(\cos(x))  = \mathcal{M}_{n,g,\A} \frac{n!}{(g)_{n}} (-2)^{n}\binom{-(2g+\A)}{n} \cos(x)^{n} \int_{-\pi}^{\pi} \frac{\d y}{2\pi} (\sin(y)^{2})^{g} C_{n}^{(g)}(\cos(y))^{2} + \ldots
\end{equation} 
It follows by definition that $\mathcal{P}_{n,g}^{(0)}(z)$ satisfies the Gegenbauer differential equation (for parameter $g+\A$) and inserting $\mathcal{M}_{n,g,\A}$ in \eqref{eq_norm_constant_4} yields that  
\begin{equation}
\mathcal{P}_{n,g+\A}^{(0)}(z) = \frac{2^{n}(g+\A)_{n}}{n!} z^{n} + \text{ l.d. },
\end{equation}
and we obtain that $\mathcal{P}_{n,g+\A}^{(0)}(z) =C_{n}^{(g+\A)}(z)$.
\end{proof}

\section{Final remarks}\label{section_conclusions}

\subsection{Irregular solutions}
We showed that the solutions of the non-stationary Lam\ee equation \eqref{eq_non_stationary_lame_defining} can be constructed with the use of explicit integral transform operators (see \eqref{eq_integral_transform_1} and \eqref{eq_integral_transform_4}), by the recursive scheme in \eqref{eq_integral_transform_general} for a suitable seed function $\psi_{n,\gn}$. The seed function, for $\gn \neq 0 , 1$, can be constructed using the series solutions in \cite{AtaiLangmann2016_2} or by, what we refer to as, the irregular solutions: The non-stationary Lam\ee equation in \eqref{eq_non_stationary_lame_defining} is clearly\footnote{It follows from the simple relation $g(g-1) = (1 - g ) ( 1-g - 1)$.} invariant under $g \to 1-g$. The solutions of \eqref{eq_non_stationary_lame_defining} for parameters $(\A, 1-g)$, with $\A$ arbitrary, are then of the form
\begin{equation}
(\theta_{1}(x)^{2})^{\half(1-g)} \mathcal{P}_{n,1-g}(\cos(x)),
\label{eq_irregular_solutions}
\end{equation}
where $\mathcal{P}$ as in \eqref{eq_ns_lame_polynomial_expansion} for suitable choices of integer $n$.\footnote{This can be inferred from known relations of  the Gegenbauer polynomials with negative index, \eg $C_{n}^{(1-g)}(z)$ is known to vanish identically if $g \in \N$ and $n \geq g$.} We refer to the eigenfunctions in Eq. \eqref{eq_irregular_solutions} as \emph{irregular} as they are not normalizable for $g>3/2$. The kernel functions can be used to transform the irregular eigenfunctions into the same form as \eqref{eq_ns_lame_eigenfunctions}: consider
\begin{equation}
\int_{\mathcal{C}'} \frac{\d y}{2\pi} \frac{(\theta_{1}(x)^{2})^{\half(g+\A)}(\theta_{1}(y)^{2})^{\half g}}{(\theta_{1}(\half(x+y))\theta_{1}(\half(x-y)))^{2g+\A}} (\theta_{1}(y)^{2})^{\half(1-g)} \mathcal{P}_{n,1-g}(\cos(y)),
\label{eq_irregular_solution_integral_transform_0}
\end{equation}
for some integration path $\mathcal{C}'$ (see also below). We note that \eqref{eq_irregular_solution_integral_transform_0} is equivalent to 
\begin{multline}
(\theta_{1}(x)^{2})^{g+\A} \oint\limits_{\gamma} \frac{\d \xi}{2\pi} \frac{1}{(\xi - \cos(x))^{2g+\A}} \\ \times \prod_{l\in\N} \frac{(1-q^{2l} ( 2\xi^{2} - 1 ) + q^{4l})}{(1 - 4 \xi \cos(x) (q^{2l}+ q^{6l}) + 2(2 \cos(x)^{2} + 2 \xi^{2}-1)q^{4l} + q^{8l})^{2 g+\A}
} \mathcal{P}_{n,1-g}(\xi)
\label{eq_irregular_solution_integral_transform}
\end{multline}
by a simple change of variables. If $(\A,g)\in \N$ then we can take take the integration contour $\gamma$ in \eqref{eq_irregular_solution_integral_transform}, as a simple closed contour around $\cos(x)$ (taken counter-clockwise) and obtain the relation
\begin{equation}
\mathcal{P}_{n,g+\A}(z) = \left. -\imag \frac{\partial^{2g+\A-1}}{\partial \xi^{2g+\A-1}} \prod_{l\in\N} \frac{(1-q^{2l} ( 2\xi^{2} - 1 ) + q^{4l})  \mathcal{P}_{n,1-g}(\xi)}{(1 - 4 \xi z (q^{2l}+ q^{6l}) + 2(2 z^{2} + 2 \xi^{2}-1)q^{4l} + q^{8l})^{2 g+\A} } \right|_{\xi=z}
\end{equation}
(up to multiplication by some normalization constant) with $z=\cos(x)$, by Cauchy's integral formula. The functions constructed in \eg \cite{etingof1994,FelderVarchenko,BazhanovMangazeev:1} can then be transformed into seed functions for our recursive scheme.

It is also interesting to consider the transformation of the kernel function under $g\to1-g$: we express the generalized kernel function identity as
\begin{equation}
\Bigl( \frac{2\imag}{\pi} ( \tilde{g}- g) \frac{\partial}{\partial \tau} + H(x;\tilde{g}) - H(y; g) - \tilde{C}_{\mu} \Bigr) \frac{(\theta_{1}(x)^{2})^{\half \tilde{g}}(\theta_{1}(y)^{2})^{\half {g}}}{\Bigl( \theta_{\mu}(\half ( x + y )) \theta_{\mu}( \half( x-y))\Bigr)^{\tilde{g} + g}} = 0 \quad( g , \tilde{g} \in \C).
\label{eq_kernel_function_identity_mu_different_coupling}
\end{equation}
It is clear that the generalized kernel function identity in \eqref{eq_kernel_function_identity_mu_different_coupling} will not be invariant under $g \to 1 - g $, nor $\tilde{g} \to 1- \tilde{g}$, as this changes the $\A = ( \tilde{g} - g )$ parameter by
\begin{subequations}\label{eq_negative_A_parameter}
\begin{align}
\A \to 2g + \A -1  , \quad k_{\mu}(x,y) \to \frac{(\theta_{1}(x)^{2})^{\half(g+\A)} (\theta_{1}(y)^{2})^{\half (1-g)}}{( \theta_{\mu}(\half(x+y))\theta_{\mu}(\half(x-y)))^{\A + 1}} ,\quad &\text{ if } g \to 1- g, \\
\A \to -2 g - \A + 1 , \quad k_{\mu}(x,y)\to \frac{(\theta_{1}(x)^{2})^{\half(1-g-\A)} (\theta_{1}(y)^{2}){\half g}}{( \theta_{\mu}(\half(x+y))\theta_{\mu}(\half(x-y)))^{ 1 - \A }} \quad &\text{ if } \tilde{g}\to 1- \tilde{g}, \\
\A \to - \A, \quad k_{\mu}(x,y) \to \frac{(\theta_{\mu}(\half(x+y))\theta_{\mu}(\half(x-y)))^{2 g + \A - 2}}{(\theta_{1}(x)^{2})^{\half(g + \A -1)} (\theta_{1}(y)^{2})^{\half(g-1)}}\quad &\text{ if } g \to 1- g \text{ and } \tilde{g}\to 1 - \tilde{g}.
\end{align}\end{subequations} 
which allows us to construct other integral transforms for many different types of solutions. In particular, the integral transform, with \eqref{eq_negative_A_parameter} as integral kernel, seems especially suitable for constructing integral representations of particular solutions considered in \cite{Fateev2009}.

\subsection{Lam\ee functions}
The Lam\ee equation corresponds to the special case of \eqref{eq_non_stationary_lame_defining}with the parameter $\A=0$, \ie
\begin{equation}
\Bigl( - \frac{\partial^{2}}{\partial x^{2}} + g(g-1)\wp(x) \Bigr) \psi(x) = E \psi(x)
\label{eq_lame_equation}
\end{equation}
Setting the parameter $\A=0$ in the integral transforms yields the integral equations
\begin{equation}
( \mathcal{K}_{g,g} \psi )(x) = \lambda \psi(x)
\end{equation}
with $\mathcal{K}=K$ in \eqref{eq_integral_transform_1} or $\mathcal{K}=\mathfrak{K}$ in \eqref{eq_integral_transform_4}, and $\lambda$ some non-zero, finite constant. It is clear that the Lam\ee functions are eigenfunctions of the integral operators. In the case of the integral transform in \eqref{eq_integral_transform_4} yields an Hilbert-Schmidt (HS) operator. In addition; if $\tau = \imag t$, with $ t>0$, then \eqref{eq_integral_transform_4} is a self-adjoint HS operator with a complete set of eigenfunctions, as was originally shown in \cite{ruijsenaars2009}.
We also believe that it is possible to construct the Lam\ee functions from the recursive integral representations of \eqref{eq_non_stationary_lame_defining} by consider the limit $\A \to 0$,$ \numb \to \infty$, such that $g = \numb \A$ is kept finite.

We note that the the Lam\ee differential operator $H$ in \eqref{eq_Lame_operator} is known to be a (formally) self-adjoint operator if $\Re(\tau)=0$. It can be shown that the non-stationary Lam\ee equation reduces to a time dependent Schr{\"o}dinger equation in the case $\tau = \imag t , \A = \imag k $ with $t >0,k\neq0$, \ie
\begin{equation}
-\imag k\frac{\partial \psi}{\partial t} = (H - E) \psi.
\end{equation}
The solutions of the non-stationary Lam\ee equation are thus the wavefunction of a quantum system determined by the Lam\ee equation.
The results in Section \ref{section_integer couplings} do not work for complex parameter $\A$, but it could be possible to analytically continue the integral transform in \ref{section_arbitrary_couplings} for complex values of $\A$. We leave this question for future work but we wish to point out that the series solutions in \cite{AtaiLangmann2016_2} is particularly suited for complex model parameters.

\subsection{Conclusions} 
We presented recursive schemes, based on kernel functions, allowing to construct  many different solutions of the non-stationary Lam\ee equation in \eqref{eq_non_stationary_lame_defining}. We have also presented explicit integral representations of the solutions for special values of the model parameters. The non-stationary Lam\ee equation is the simplest non-trivial case of many other interesting models which can be solved by a similar recursive scheme.

As mentioned in the introduction, our goal is to give a complementary method to series solutions of the non-stationary Heun equation in \cite{AtaiLangmann2016_2}. The non-stationary Heun equation, which can be written as
\begin{equation}
\Bigl( \frac{\imag }{\pi} \A \frac{\partial}{\partial \tau} - \frac{\partial^{2}}{\partial x^{2}} + \sum_{\nu=0}^{3} g_{\nu}( g_{\nu}-1) \wp(x + \omega_{\nu} | \pi , \pi\tau ) \Bigr) \psi = E(\tau) \psi
\label{eq_non_stationary_heun_defining}
\end{equation}
with $\omega_{0}\equiv 0 , \omega_{1} = \pi , \omega_{2} \equiv -\pi - \pi \tau , \omega_{3}\equiv \pi\tau$, and model parameters $(\A, \{g_{\nu}\}_{\nu=0}^{3})$, also has explicitly known kernel functions \cite{Langmann_Takemura}. We believe that the recursive scheme in this paper can be adapted to the non-stationary Heun equation and that the Lemmas in Sections \ref{section_proof_of_theorem_1} and \ref{section_proof_of_theorem_2} can be generalized for the non-stationary Heun equation.

There are also explicitly known kernel functions for the many-variable generalizations of the non-stationary Lam\ee equation, \ie the non-stationary elliptic Calogero-Sutherland model \cite{OlshanetskyPerelomov,Langmann_eCS_source} and the non-stationary Inozemtsev model \cite{Inozemtsev,Langmann_Takemura}. We believe that it is also possible to construct integral representations of solutions of these many-variable generalizations of the non-stationary Lam\ee equation.

One important question is in regard to the uniqueness of the solutions of the non-stationary Lam\ee equation: Our results show that we can translate this question to the question regarding the uniqueness of the integral transform. We leave this important question for future work.

We finally note that the integral relations are useful for determining interesting properties of the solutions \cite{LAMBE,Erdelyi,Kazakov1996,Novikov2006}. We mention in particular, a paper by Ruijsenaars that found a hidden permutation symmetry of the solutions to the Heun equation \cite{ruijsenaars2009}. We hope that our results can be used to determine properties of solutions for the non-stationary Lam\ee equation.

\begin{acknowledgments}
I would like to thank E. Langmann, H. Rosengren, and J. Lidmar for helpful comments on the manuscript and discussions. I am also grateful to O. Palm and P. Moosavi for discussions.
\end{acknowledgments}

\appendix

\section{Special functions}\label{appendix_special_functions}

 For the convenience of the reader, we collect here definitions and properties of the special functions that are used in this paper. We use the standard results and notations that can be found in \eg \cite{WhitWat} or \cite{NIST:DLMF}.

The (raising) Pochhammer symbol $(x)_{n}$ is defined as 
\begin{equation}
(x)_{n} \equiv x ( x + 1 ) \cdots (x + n -1 )
\label{eq_Pochhammer_def}
\end{equation}
for any positive integer $n$. It is also possible to express the Pochhammer symbol as $(x)_{n} = \Gamma( x+ n )  / \Gamma(x)$, which is valid for any $x,n\in\C$ such that they do not hit the poles of the Euler $\Gamma$-function.
The binomial coefficients can be expressed in terms of the Pochhammer symbols, given by
\begin{equation}
\binom{x}{n} = \frac{(x - n + 1)_{n}}{n!} = \frac{\Gamma( x + 1 ) }{\Gamma(x - n + 1) \Gamma(n+1)}.
\end{equation}
(See \S 5 of \cite{NIST:DLMF}.)

The Gegenbauer polynomials are the unique normalizable solution of the differential equation
\begin{equation}
\Bigl( (1-z^{2}) \frac{\partial^{2}}{\partial z^{2}} - (2g+1) z \frac{\partial }{\partial z} + n( n+2g) \Bigr) C_{n}^{(g)}(z) = 0 \quad (n\in\N_{0}),
\label{eq_Gegenbauer_diff_equation}
\end{equation}
and has the explicit series representation 
\begin{equation}
C_{n}^{(g)}(z)  = \sum_{k=0}^{ \lfloor n / 2 \rfloor} \frac{(-1)^{k} ( g )_{n - k}}{k! (n-2k)!} (2z)^{n-2k}.
\label{eq_Gegenbauer_series_rep}
\end{equation}
The Gegenbauer polynomials $C_{n}^{(g)}$ satisfy the three-term recurrence relation
\begin{equation}
z C_{n}^{(g)}(z) = a_{n} C_{n+1}^{(g)}(z) - c_{n}C_{n-1}^{(g)}(z) , \quad a_{n}\equiv \frac{n+1}{2(n+g)} ,\quad c_{n} \equiv \frac{n+2g-1}{2(n+g)}.
\label{eq_three_term_rec}
\end{equation}
The Gegenbauer polynomials have the generating function 
\begin{equation}
\frac{1}{(1 - 2 z \xi + \xi^{2} )^{g}} = \sum_{k\in\N_{0}} C_{k}^{(g)}(z) \xi^{k},
\label{eq_Gegenbauer_generating_function}
\end{equation}
which yields the integral representation of the Gegenbauer polynomials
\begin{equation}
C_{n}^{(g)}(z) = \oint_{| \xi | < 1} \frac{\d \xi}{2\pi \imag \xi } \frac{1}{(1-2 z \xi + \xi^{2})^{g}} \xi^{-n}.
\end{equation}

The Gegenbauer polynomials are known to satisfy the orthogonality condition
\begin{equation}
\int_{-\pi}^{\pi} \frac{\d y}{2\pi} (\sin(x)^{2})^{g} C_{n}^{(g)}(\cos(y)) C_{m}^{(g)}(\cos(y)) = \delta_{n,m} \mathfrak{h}_{n}
\label{eq_Gegenbauer_norm_1}
\end{equation}
with 
\begin{equation}
\mathfrak{h}_{n} \equiv \frac{2^{1-2g} \Gamma(n + 2 g)}{n! (n+g)  [\Gamma(g)]^{2}},
\label{eq_Gegenbauer_norm_2}
\end{equation}
and where the integral on the l.h.s. of \eqref{eq_Gegenbauer_norm_1} is proportional to the inner product $\langle \cdot , \cdot \rangle$ by a simple change of variables.

(See \S18 \cite{NIST:DLMF}.)

\subsection{Elliptic functions}\label{appendix_elliptic_functions}

The Weierstrass $\wp$-function with periods $(2\wone, 2 \wthr)$ is defined as
\begin{equation}
\wp(x | \wone , \wthr) \equiv \frac{1}{x^{2}} + \sum\limits_{ (n,m)\in \Z^{2}\setminus\{(0,0)\}} \frac{1}{(x - \Omega_{n,m})^{2}} - \frac{1}{\Omega_{n,m}^{2}}
\label{eq_Weierstrass_P}
\end{equation}
with the lattice $\Omega_{n,m}\equiv 2 n \wone + 2 m \wthr$. For simplicity we set the periods to
\begin{equation}
2\wone = \pi , \quad 2  \wthr = \pi \tau
\end{equation}
in the following and always use $\wp(x)$ as a short-hand for $\wp(x | \phalf , \phalf \tau)$. We also recall the \emph{nom\ee} $q \equiv \e^{\imag \pi \tau}$.
The Weierstrass $\wp$-function can also be expanded as
\begin{equation}
\wp(x) = \frac{1}{\sin(x)^{2}} - \frac{\eta_{1}}{\pi} - 8 \sum_{n\in\N} \frac{n q^{2n}}{1-q^{2n}} \cos( 2 n x ) \quad ( x \in [-\pi , \pi ])
\label{eq_Weierstrass_P_expansion}
\end{equation}
where the sum is absolutely convergent.

The Jacobi theta functions are defined as 
\begin{equation}
\begin{split}
\vartheta_{1}(x) \equiv 2 \sum_{n\in\N_{0}} (-1)^{n} q^{(n+\half)^{2}} \sin((2n+1)x) &, \quad \vartheta_{2}(x) \equiv 2 \sum_{n\in\N_{0}} q^{(n+\half)^{2}} \cos((2n+1)x) \\
\vartheta_{3}(x) \equiv 1 + 2 \sum_{n\in\N}q^{n^{2}} \cos(2 n x) &, \quad \vartheta_{4}(x) \equiv 1 + 2 \sum_{n\in\N} (-1)^{n} q^{n^{2}} \cos(2nx)
\end{split} 
\end{equation}
It follows from the definitions that $\vartheta_{\nu}$ has simple zeroes at the points $z_{\nu}$ ($\nu=1,2,3,4$), given by \cite{WhitWat}
\begin{equation}
z_{1} = 0 \ (\text{mod }(\pi , \pi \tau)), \quad z_{2} = \phalf \ (\text{mod }(\pi , \pi \tau)), \quad z_{3} = \phalf ( 1 + \tau)  \ (\text{mod }(\pi , \pi \tau)), \quad z_{4} = \phalf \tau \ (\text{mod }(\pi , \pi \tau)).
\end{equation}
The relations to the $\theta_{\nu}$-functions (see Section \ref{section_notation}) follows from the Jacobi triple product identity.

The Jacobi $\vartheta$-functions all satisfy the differential equation
\begin{equation}
\Bigl( \frac{4 \imag }{\pi} \frac{\partial}{\partial \tau} - \frac{\partial^{2}}{\partial x^{2}} \Bigr) \vartheta_{\nu}(x)=0 \quad (\nu = 1,2,3,4).
\label{eq_Jacobi_heat_equation}
\end{equation}

It follows that the $\theta_{\nu}$-functions in \eqref{eq_Jacobi_theta_functions} satisfy
\begin{equation}
\Bigl( \frac{4 \imag }{\pi} \frac{\partial}{\partial \tau} - \frac{\partial^{2}}{\partial x^{2}} \Bigr) \theta_{\nu}(x) = \begin{cases} \theta_{\nu}(x), \quad \text{ if } \nu = 1,2 \\ 
0\quad \text{   if } \nu = 3,4
\end{cases}
\label{eq_theta_heat_equation}
\end{equation}

\subsubsection{Relations}\label{appendix_relations}
If follows from a simple check that 
\begin{equation}
\vartheta_{1}'(0) = 2 q^{\quarter} G^{3},
\end{equation}
with $G$ in \eqref{eq_G_constant}.
We also have the relations
\begin{equation}
\frac{\eta_{1}}{\pi} = \frac{1}{12} \frac{\vartheta_{1}'''(0)}{\vartheta_{1}'(0)} = \frac{1  }{3} \frac{\imag}{\pi} \frac{1}{\vartheta_{1}'(0)}\frac{\partial \vartheta_{1}'(0)}{\partial \tau}
\end{equation}

The Weierstrass $\wp$-function can be expressed as
\begin{equation}
\wp(x ) =  \left( \frac{\theta_{1}'( x)}{\theta_{1}(x)} \right)^{2} - \frac{\theta_{1}''( x)}{\theta_{1}( x)} - \frac{4\eta_{1}}{\pi}
\label{eq_Weierstrass_P_theta_def}
\end{equation}
with $$\theta_{1}'( x) \equiv \frac{\partial}{\partial x} \theta_{1}( x) , \quad \theta_{1}''(x ) \equiv \frac{\partial^{2}}{\partial x^{2}} \theta_{1}( x). $$

The theta functions also satisfies 
\begin{equation}
\Bigl( \frac{\theta_{1}'(x_{1})}{\theta_{1}(x_{1})} +  \frac{\theta_{1}'(x_{2})}{\theta_{1}(x_{2})} +  \frac{\theta_{1}'(x_{3})}{\theta_{1}(x_{3})} \Bigl)^{2} = \wp(x_{1}) + \wp(x_{2}) + \wp(x_{3}) \quad ( x_{1} + x_{2} + x_{3} =0 ),
\label{eq_Weierstrass_P_theta_relation_1}
\end{equation}
which yields
\begin{equation}
\begin{split}
 \frac{\theta_{1}'(x_{1})\theta_{1}'(x_{2})}{\theta_{1}(x_{1})\theta_{1}(x_{2})} + \frac{\theta_{1}'(x_{1})\theta_{1}'(x_{3})}{\theta_{1}(x_{1})\theta_{1}(x_{3})} + \frac{\theta_{1}'(x_{2})\theta_{1}'(x_{3})}{\theta_{1}(x_{2})\theta_{1}(x_{3})} &= \frac{1}{2} ( \wp(x_{1}) +\wp(x_{2}) + \wp(x_{3})) - \frac{1}{2}\Bigl(\frac{\theta_{1}'(x_{1})^{2}}{\theta_{1}(x_{1})^{2}} +  \frac{\theta_{1}'(x_{2})^{2}}{\theta_{1}(x_{2})^{2}} +  \frac{\theta_{1}'(x_{3})^{2}}{\theta_{1}(x_{3})^{2}}\Bigr) \\ 
&= - \frac{1}{2} \Bigl( \frac{\theta''(x_{1})}{\theta_{1}(x_{1})} + \frac{\theta''(x_{2})}{\theta_{1}(x_{2})} + \frac{\theta''(x_{2})}{\theta_{1}(x_{2})} + 12 \frac{\eta_{1}}{\pi} \Bigr)
\end{split}
\label{eq_Weierstrass_P_theta_relation_2}
\end{equation}

The $\vartheta_{1}$-function is known to satisfy the (quasi-) periodicity relations
\begin{equation}
\vartheta_{1}(x + \pi n + \pi m \tau ) = (-1)^{n+m} q^{-m^{2}} \e^{-2 \imag m x} \vartheta_{1}(x), \quad n,m\in\Z
\end{equation}
which also extends to the $\theta_{1}$-function.
The $\theta_{\nu}$-functions satisfy the half-period shift relations\begin{subequations}\label{eq_theta_half_shifts}
\begin{align}
\theta_{1}( x \pm \phalf ) = \pm \theta_{2}(x) , &\quad \theta_{1}(x \pm \phalf \tau) = \pm \imag \e^{\mp \imag x} q^{-\half} \theta_{4}(x), \\
\theta_{2}( x \pm \phalf ) = \mp \theta_{1}(x) , &\quad \theta_{2}(x \pm \phalf \tau) = \e^{\mp \imag x} q^{-\half} \theta_{3}(x), \\
\theta_{3}( x \pm \phalf ) = \pm \theta_{4}(x) , &\quad \theta_{3}(x \pm \phalf \tau) = \e^{\mp \imag x} \theta_{2}(x), \\
\theta_{4}( x \pm \phalf ) = \pm \theta_{3}(x) , &\quad \theta_{4}(x \pm \phalf \tau) = \pm\imag\e^{\mp \imag x}  \theta_{1}(x), \\
\end{align}
\end{subequations}
which follows from the well-known half-period shift relations of the Jacobi $\vartheta_{\nu}$-functions (see \eg Chapter XXI of \cite{WhitWat}).

\begin{align}
\vartheta_{1}(y) &= \frac{-\imag q^{\quarter}}{2} \e^{-\imag y} ( 1 - \e^{2 \imag y}) \prod\limits_{l \in\N} ( 1 - q^{2 l} \e^{2 \imag y})(1-q^{2 l} \e^{-2\imag y}) \\
&= \frac{-\imag q^{\quarter}}{2} \e^{-\imag y} \Theta_{1}(\e^{2 \imag y})
\end{align}
\begin{align}
\vartheta_{1}(\half(x+y)) \vartheta_{1}(\half(x-y)) = 2 q^{\half} \e^{\imag y } \Theta(\cos(x) , \e^{\imag y})
\end{align}
which follows by a straightforward check from the definitions in Section \ref{section_notation}.

\subsubsection{Trigonometric limit}\label{appendix_trig_limit}

The Weierstrass $\wp$-function has a well-known \cite{WhitWat} trigonometric limit, given by
\begin{equation}
\wp(x) \to \frac{1}{\sin(x)^{2}} - \frac{1}{3} \quad \text{as} \quad \Im(\tau) \to +\infty.
\label{eq_trig_limit_Weierstrass}
\end{equation}
(Recall that $\wp(x) \equiv \wp(x | \half \wone , \half \wone \tau)$.)

It is clear from the definitions in \eqref{eq_Jacobi_theta_functions} that
\begin{equation}
\begin{split}
\theta_{1}(x) \to \sin(x) ,& \quad  \theta_{2} \to \cos(x) \\
\theta_{3}(x) \to 1 , & \quad \theta_{4}(x) \to 1
\end{split}
\label{eq_trig_limit_theta}
\end{equation}
uniformly as $\Im(\tau)\to+\infty$.

\section{Proof of Lemmas \ref{lemma_gen_KF_id_1} and \ref{lemma_gen_KF_id_4}}\label{appendix_proof_of_lemmas}
In this Section it is proven that the functions
\begin{equation}
k_{\mu}(x,y) = \frac{(\theta_{1}(x)^{\half})^{\half (g+\A) } (\theta_{1}(y)^{2})^{\half g}}{(\theta_{\mu}(\half(x+y)) \theta_{\mu}(\half(x-y)))^{2g+\A}},
\label{eq_appendix_kernel_functions_mu}
\end{equation}
with $\mu\in \{1,2,3,4\}$ fixed, satisfies the generalized kernel function identity
\begin{equation}
\Bigl( \frac{2\imag}{\pi} \A \frac{\partial}{\partial \tau} + H_{\L}(x;g+\A) - H_{\L}(y;g) - C_{\mu} \Bigr) k_{\mu}(x,y) = 0
\end{equation}
with 
\begin{equation}
C_{\mu}  = \begin{cases}
4\A ( 1 - 2g - \A) \frac{\eta_{1}}{\pi}+ 6\A(2g+\A)\left( \frac{\eta_{1}}{\pi}-\frac{1}{12}\right) , \quad \text{ if } \mu = 1, 2 \\
\A (2g+\A) +4\A ( 1 - 2g - \A) \frac{\eta_{1}}{\pi}+ 6\A(2g+\A)\left( \frac{\eta_{1}}{\pi}-\frac{1}{12}\right)  ,\quad \text{ if } \mu = 3 ,4
\end{cases}
\end{equation}

We start with the $\mu=1$ case which  is, as stated after Lemma \ref{lemma_gen_KF_id_1}, a special case of  Corollary 3.2 of \cite{Langmann_Takemura}. 

It follows from straightforward calculations that
\begin{equation}
\Bigl( -\frac{\partial^{2}}{\partial x^{2}} + \frac{\partial^{2}}{\partial y^{2}} \Bigr) k_{1}(x,y) =  \mathcal{W}(x,y) k_{1}(x,y),
\end{equation}
with 
\begin{equation}
\mathcal{W}(x,y) = \mathcal{W}_{1}(x,y) + \mathcal{W}_{2}(x,y)
\end{equation}
where
\begin{align}
\mathcal{W}_{1} &= (g+\A)  \Bigl( \frac{\theta_{1}'(x)^{2}}{\theta_{1}(x)^{2}} - \frac{\theta_{1}''(x)}{\theta_{1}(x)}\Bigr)  - (g+\A)^{2}\frac{\theta_{1}'(x)^{2}}{\theta_{1}(x)^{2}} - g\Bigl( \frac{\theta_{1}'(y)^{2}}{\theta_{1}(y)^{2}} - \frac{\theta_{1}''(y)}{\theta_{1}(y)} \Bigr) + g^{2} \frac{\theta_{1}'(y)^{2}}{\theta_{1}(y)^{2}} \\
\mathcal{W}_{2} &= (g+\A)(2g+\A)\frac{\theta_{1}'(x)}{\theta_{1}(x)} \Bigl( \frac{ \theta_{1}'(\half(x+y))}{\theta_{1}(\half(x+y))} + \frac{ \theta_{1}'(\half(x-y))}{\theta_{1}(\half(x-y))} \Bigr) \\ &-  g (2g+\A) \frac{\theta_{1}'(y)}{\theta_{1}(y)} \Bigl( \frac{ \theta_{1}'(\half(x+y))}{\theta_{1}(\half(x+y))} - \frac{\theta_{1}'(\half(x-y))}{\theta_{1}(\half(x-y))} \Bigr)  -(2g+\A)^{2} \frac{\theta_{1}'(\half(x+y)) \theta_{1}'(\half(x-y))}{\theta_{1}(\half(x+y))\theta_{1}(\half(x-y))}.
\end{align}

Using \eqref{eq_Weierstrass_P_theta_def} yields
\begin{equation}
\mathcal{W}_{1} = - (g+\A)(g+ \A - 1) (\wp(x)+ 4\frac{\eta_{1}}{\pi}) + g(g-1) (\wp(y) + 4 \frac{\eta_{1}}{\pi} ) - (g+\A)^{2} \frac{\theta_{1}''(x)}{\theta_{1}(x)} + g^{2} \frac{\theta_{1}''(y)}{\theta_{1}(y)}.
\end{equation}
Equation \eqref{eq_Weierstrass_P_theta_relation_1} yields that
\begin{equation}
\begin{split}
\frac{\theta_{1}'(x)}{\theta_{1}(x)} \Bigl( \frac{ \theta_{1}'(\half(x+y))}{\theta_{1}(\half(x+y))} + \frac{ \theta_{1}'(\half(x-y))}{\theta_{1}(\half(x-y))} \Bigr) =  \frac{\theta_{1}'(\half(x+y)) \theta_{1}'(\half(x-y))}{\theta_{1}'(\half(x+y)) \theta_{1}'(\half(x-y))} \\ - \frac{1}{2} ( \wp(x) + \wp(\half(x+y)) + \wp(\half(x-y)) ) + \frac{1}{2} \Bigl( \frac{\theta_{1}'(x)^{2}}{\theta_{1}(x)^{2}} + \frac{\theta_{1}'(\half(x-y))^{2}}{\theta_{1}(\half(x-y))^{2}}+ \frac{\theta_{1}'(\half(x+y))^{2}}{\theta_{1}(\half(x+y))^{2}} \Bigr),
\end{split}
\end{equation}
(using $x_{1} = - x, x_{2} = \half(x-y), x_{3}= \half(x+y)$,) and
\begin{equation}
\begin{split}
 \frac{\theta_{1}'(y)}{\theta_{1}(y)} \Bigl( \frac{ \theta_{1}'(\half(x+y))}{\theta_{1}(\half(x+y))} - \frac{\theta_{1}'(\half(x-y))}{\theta_{1}(\half(x-y))} \Bigr) = - \frac{\theta_{1}'(\half(x+y)) \theta_{1}'(\half(x-y))}{\theta_{1}'(\half(x+y)) \theta_{1}'(\half(x-y))} \\ - \frac{1}{2} ( \wp(y) + \wp(\half(x+y)) + \wp(\half(x-y)) ) + \frac{1}{2} \Bigl( \frac{\theta_{1}'(y)^{2}}{\theta_{1}(y)^{2}} + \frac{\theta_{1}'(\half(x-y))^{2}}{\theta_{1}(\half(x-y))^{2}}+ \frac{\theta_{1}'(\half(x+y))^{2}}{\theta_{1}(\half(x+y))^{2}} \Bigr),
\end{split}
\end{equation}
(using $x_{1}=-y, x_{2} = \half(y+x), x_{3}=\half(y-x)$).
If follows from \eqref{eq_Weierstrass_P_theta_def} that
\begin{equation}
\begin{split}
\mathcal{W}_{2}(x,y) &= \\& (g+\A)(2g+\A)\Bigl( \frac{\theta_{1}'(\half(x+y)) \theta_{1}'(\half(x-y))}{\theta_{1}'(\half(x+y)) \theta_{1}'(\half(x-y))} + \frac{1}{2} \Bigl( \frac{\theta_{1}''(x)}{\theta_{1}(x)} + \frac{\theta_{1}''(\half(x+y))}{\theta_{1}(\half(x+y))}+\frac{\theta_{1}''(\half(x-y))}{\theta_{1}(\half(x-y))} + 12 \frac{\eta_{1}}{\pi} \Bigr) \Bigr) \\
&-g(2g+\A)\Bigl(  - \frac{\theta_{1}'(\half(x+y)) \theta_{1}'(\half(x-y))}{\theta_{1}'(\half(x+y)) \theta_{1}'(\half(x-y))} + \frac{1}{2} \Bigl( \frac{\theta_{1}''(y)}{\theta_{1}(y)} + \frac{\theta_{1}''(\half(x+y))}{\theta_{1}(\half(x+y))}+\frac{\theta_{1}''(\half(x-y))}{\theta_{1}(\half(x-y))} +12 \frac{\eta_{1}}{\pi} \Bigr) \Bigr) \\
&-(2g+\A)^{2} \frac{\theta_{1}'(\half(x+y)) \theta_{1}'(\half(x-y))}{\theta_{1}(\half(x+y))\theta_{1}(\half(x-y))} \\
&= \frac{\A(2g+\A)}{2} \Bigl(\frac{\theta_{1}''(\half ( x+ y))}{\theta_{1}(\half ( x+ y))} + \frac{\theta_{1}''(\half ( x - y))}{\theta_{1}(\half ( x - y))} \Bigr) + \frac{1}{2} ( g + \A)(2g+\A) \frac{\theta_{1}''(x)}{\theta_{1}(x)} \\ & - \frac12  g(2g + \A) \frac{\theta_{1}''(y)}{\theta_{1}(y)} + 6 \A(2g+\A) \frac{\eta_{1}}{\pi}.
\end{split}
\end{equation}

Combining these relations yields that
\begin{equation}\begin{split}
(H(x;g+\A) - H(y;g) ) k_{1}(x,y) =  2\A(2 -  10 g - 5 \A)\frac{\eta_{1}}{\pi} k_{1}(x,y) \\ + \Bigl( - \half \A( g+ \A) \frac{\theta_{1}''(x)}{\theta_{1}(x)} - \half \A g \frac{\theta_{1}''(y)}{\theta_{1}(y)}  +\half \A ( 2 g + \A) \Bigl(\frac{\theta_{1}''(\half ( x+ y))}{\theta_{1}(\half ( x+ y))} + \frac{\theta_{1}''(\half ( x - y))}{\theta_{1}(\half ( x - y))}  \Bigr) \Bigr) k_{1}(x,y),
\end{split}\end{equation}
with $H(x;g)$ the Lam\ee differential operator in \eqref{eq_Lame_operator}.
Finally we note that \eqref{eq_theta_heat_equation} yields that 
\begin{multline}
\frac{2\imag}{\pi} \A \frac{\partial}{\partial \tau} k_{1}(x,y) = \Bigl(  \frac12  \A( g+ \A) \frac{\theta_{1}''(x)}{\theta_{1}(x)} + \half \A g \frac{\theta_{1}''(y)}{\theta_{1}(y)}  \Bigr. \\ \Bigl. - \frac12 \A ( 2 g + \A) \Bigl(\frac{\theta_{1}''(\half ( x+ y))}{\theta_{1}(\half ( x+ y))} + \frac{\theta_{1}''(\half ( x - y))}{\theta_{1}(\half ( x - y))}  \Bigr) -\half \A  (2g+\A) \Bigr) k_{1}(x,y)
\end{multline}

Collecting the results yields that 
\begin{multline}
\Bigl( \frac{2\imag}{\pi}\A \frac{\partial}{\partial \tau} + H(x;g+\A) - H(y; g) \Bigr)k_{1}(x,y) =  ( 4 \A( 1- 2 g - \A) \frac{\eta_{1}}{\pi}  + 6\A ( 2g + \A) ( \frac{\eta_{1}}{\pi} - \frac{1}{12})) k_{1}(x,y),
\end{multline}
which concludes the proof for the $\mu=1$ case.

Note that the $\mu=2$ and $\mu=3$ cases follows from the $\mu=1$ and $\mu=4$ cases, respectively, by a simple shift $y \to y+\pi$ (note that the Lam\ee differential operator is invariant under this shift).  For the $\mu=4$ case we note that if $k_{1}(x,y)$ satisfies the kernel function identity in \eqref{eq_kernel_function_identity}, then
\begin{equation}
\Bigl( \frac{2\imag}{\pi} \A \frac{\partial}{\partial \tau} + H(x;g+\A) - H(y;g) - 2 \imag \A \frac{\partial}{\partial y} - C_{1} \Bigr) k_{1}(x , y + \pi \tau).
\end{equation}
Using \eqref{eq_theta_half_shifts} yields that 
\begin{equation}
\begin{split}
k_{1}(x,y+\pi\tau) &= \frac{(\theta_{1}(x)^{2})^{\half(g+\A)} (\theta_{1}(y + \pi \tau)^{2})^{\half g}}{( \theta_{1}(\half (x + y ) + \phalf \tau) \theta_{1}(\half(x-y) - \phalf \tau))^{2g+\A}} = \frac{(\theta_{1}(x)^{2})^{\half(g+\A)}\Bigl( (- \e^{-2\imag y} q^{-1} \theta_{1}(y)^{2} \Bigr)^{\half g}}{( \e^{-\imag y} q^{-1} \theta_{1}(\half (x + y )) \theta_{1}(\half(x-y)))^{2g+\A}} \\ 
&= \e^{\imag \A y} q^{ g +\A} k_{4}(x,y)
\end{split}
\end{equation}
with $k_{4}(x,y)$ in \eqref{eq_appendix_kernel_functions_mu}. Finally we note that 
\begin{equation}
(\e^{\imag \A y} q^{g+\A})^{-1} \Bigl( \frac{2\imag}{\pi} \A \frac{\partial}{\partial \tau} + \frac{\partial^{2}}{\partial y^{2}} - 2 \imag \A \frac{\partial}{\partial y} \Bigr) \e^{\imag \A y} q^{g+\A} = \Bigl( \frac{2\imag}{\pi} \A \frac{\partial}{\partial \tau} + \frac{\partial^{2}}{\partial y^{2}} - \A(2g+\A) \Bigr)
\end{equation}
which yields the generalized kernel function identity in \eqref{eq_kernel_function_4} for $\mu=4$.
(Note that the generalized kernel function identities, for $\mu=2,3,4$, can also be obtained by straightforward calculations, as in the $\mu=1$ case above.)

\end{document}